\documentclass{article}
\usepackage{rayyli}
\usetikzlibrary{spy,patterns}
\usepackage{quantikz}
\newcommand\nxt{\textsf{Next}}
\usepackage{enumerate}   
\usepackage{svg}

\title{Optimal Locality and Parameter Tradeoffs for Subsystem Codes}

\author{
Samuel Dai\thanks{Department of Physics, Northeastern University. \url{dai.sa@northeastern.edu}.} 
\and 
Ray Li\thanks{Math \& CS Department, Santa Clara University. \url{rli6@scu.edu}.}
\and 
Eugene Tang \thanks{Departments of Math and Physics, Northeastern University.  \url{e.tang@northeastern.edu.}}
}

\date{\today}

\begin{document}

\maketitle
\begin{abstract}
      We study the tradeoffs between the locality and parameters of subsystem codes. We prove lower bounds on both the number and lengths of interactions in any $D$-dimensional embedding of a subsystem code. Specifically, we show that any embedding of a subsystem code with parameters $[[n,k,d]]$ into $\bbR^D$ must have at least $M^*$ interactions of length at least $\ell^*$, where
      \[
        M^* = \Omega(\max(k,d)), \quad\text{and}\quad \ell^* = \Omega\bigg(\max\bigg(\frac{d}{n^\frac{D-1}{D}}, \bigg(\frac{kd^\frac{1}{D-1}}{n}\bigg)^\frac{D-1}{D}\bigg)\bigg).
      \]
      We also give tradeoffs between the locality and parameters of commuting projector codes in $D$-dimensions, generalizing a result of Dai and Li~\cite{dai2024locality}. 
      We provide explicit constructions of embedded codes that show our bounds are optimal in both the interaction count and interaction length.
\end{abstract}

\newpage

\section{Introduction}
\label{sec:intro}
Quantum computing necessitates the manipulation of fragile states of information. The most promising way towards large-scale fault-tolerant quantum computing involve the extensive use of quantum error-correcting codes (QECCs). Physical implementations of quantum computing hardware naturally favor architectures which are local in $2$ or $3$ spatial dimensions --- architectures where the qubits are embedded in 2 or 3 dimensions, and interactions occur only between qubits that are spatially nearby. On the other hand, it has long been known that the constraint of spatial locality places severe limitations on the parameters of QECCs. For example, the Bravyi-Terhal \cite{bravyi2009no} and Bravyi-Pouli-Terhal (BPT) \cite{bravyi2010tradeoffs} bounds state that a commuting projector code whose constraints are local in $D$-dimensions necessarily have code parameters satisfying, respectively,
\begin{align}
d= O(n^{\frac{D-1}{D}}),\qquad\text{and}\qquad kd^{\frac{2}{D-1}} = O(n).\label{eq:bpt_bound}
\end{align}
These bounds suggest that there are tradeoffs between better code performance and the cost of non-local implementation. Consequently, the \emph{locality} of a QECC becomes another key factor to consider when choosing a code for applications.

What is the quantitative tradeoff between locality and code quality? This problem was initially investigated by Baspin and Krishna \cite{baspin2022quantifying}, who asked, for a quantum low-density parity-check (qLDPC) code in $D$-dimensions, how many ``long-range'' interactions must there be, and how long must those interactions be? Baspin and Krishna gave bounds for $D$-dimensional codes which are nearly optimal in certain parameter settings. For $2$-dimensional codes, Dai and Li \cite{dai2024locality} improved the bounds to be tight across all parameter regimes and also gave matching constructions that saturate the upper bounds (see also Hong, et al.~\cite{hong2023long}, who considered the special case $k=1,d=\sqrt{n}$ for 2-dimensional codes). Dai and Li showed that an $[[n,k,d]]$ quantum code embedded in $2$ spatial dimensions must have $\Omega(M^*)$ interactions of length $\Omega(\ell^*)$, where
\begin{align}
M^* = \max(k,d),\qquad\text{and}\qquad \ell^* = \max(\frac{d}{\sqrt{n}},\sqrt[4]{\frac{kd^2}{n}}).
\end{align}
Both the interaction count $M^*$ and interaction length $\ell^*$ are tight in strong ways.

In this paper, we study the locality versus parameter tradeoffs for quantum \emph{subsystem} codes. 
Bravyi~\cite{bravyi2011subsystem} showed that the BPT bound could be violated by the use of local subsystem codes, providing $2$D-local subsystem codes with parameters $k,d=\Theta(\sqrt{n})$. Subsystem codes are nevertheless constrained by locality. Bravyi~\cite{bravyi2011subsystem} showed that a $[[n,k,d]]$ subsystem code whose gauge generators are local in a $D$-dimensional lattice embedding satisfies
\begin{equation}
    \label{eq:bravyi_bound}
    d= O(n^{\frac{D-1}{D}}),\qquad\text{and}\qquad 
    kd^{\frac{1}{D-1}} = O(n).
\end{equation}
While previous work has made it clear that outperforming local quantum codes requires copious amounts of long-ranged interactions, it is not a priori clear whether the same requirements hold for subsystem codes. Is it possible that small violations of locality in the gauge generators suffice to define subsystem codes parametrically better than those allowed by Bravyi's bound? More concretely:

\begin{question}\label{question:main}
    How much non-locality is required for a subsystem code to exceed Bravyi's bound?
\end{question}

We address Question~\ref{question:main} by demonstrating that subsystem codes, like their commuting projector counterparts, require an extensive number of long-ranged interactions to surpass Bravyi's bound. We also provide constructions of subsystem codes that show our lower bounds are tight in strong ways. Additionally, we also generalize the results of Dai and Li~\cite{dai2024locality} from $2$-dimensions to $D$-dimensions. Our work establishes optimal bounds on interaction lengths and counts for embeddings of both commuting projector codes and subsystem codes in any number of dimensions.

\subsection{Main Result} 

We study subsystem codes whose gauge generators are not necessarily local. Our main result is a lower bound on the number and length of interactions in any $D$-dimensional embedding of a $[[n,k,d]]$ subsystem code. Formally, a $D$-dimensional embedding is a mapping of the code's $n$ physical qubits into $\mathbb{R}^D$, such that any two qubits are at distance at least 1.

\begin{theorem}[Main Result for Subsystem Codes]
\label{thm:main}
For any $D\ge 2$, there exist constants $c_0=c_0(D)>0$ and $c_1=c_1(D) > 0$ such that the following is true: Any $D$-dimensional embedding of a nontrivial\footnote{Nontrivial here simply means that $k > 0$.} $[[n,k,d]]$ subsystem code with $kd^\frac{1}{D-1} \geq c_1  n$ or $d\ge c_1 n^{\frac{D-1}{D}}$ must have at least $M^*$ interactions of length $\ell^*$, where 
\begin{equation}
    M^* = c_0 \cdot \max(k,d),\qquad\text{and} \qquad \ell^* = c_0 \cdot \max\bigg(\frac{d}{n^\frac{D-1}{D}}, \bigg(\frac{kd^\frac{1}{D-1}}{n}\bigg)^\frac{D-1}{D}\bigg).
\end{equation}
\end{theorem}
Prior to this work, no such bounds of this form were known for subsystem codes aside from Bravyi's original bound. While such bounds were known for commuting projector codes, our bound shows that a locality versus parameter trade-off also holds for subsystem codes. 
Our result also generalizes Bravyi's bound \cite{bravyi2011subsystem}, not only in that we (optimally) address the number and length of long-range interactions, but also in that we handle more general embeddings. Bravyi's bound~\cite{bravyi2011subsystem} considers only embeddings onto a $n^{1/D}\times\cdots\times n^{1/D}$ lattice, but our bound applies to arbitrary embeddings, even those not constrained to a $O(n^{1/D})\times\cdots\times O(n^{1/D})$ box (see Section~\ref{sec:main-d} for further discussion).

Like for stabilizer codes, subsystem codes beyond the ``local regime'' --- above the BPT bound for stabilizer codes, or above the Bravyi bound for subsystem codes ---  need copious amounts of non-locality. 
In particular, the number of required long-range interactions $\Omega(\max(k,d))$ is the same for both subsystem and stabilizer codes. Additionally, for codes with a large number $k$ of logical qubits, the required length of the long range interactions is similar. For example, a $2$-dimensionally embedded asymptotically good subsystem code (with $k,d=\Omega(n)$) needs $M^*=\Omega(n)$ interactions of length $\ell^*=\Omega(\sqrt{n})$ -- the worst possible case -- just as for stabilizer codes. Our results show that, compared to stabilizer codes, subsystem codes do not offer substantial improvements in locality outside of the ``local regime,'' though they can offer some quantitative improvements in the interaction length.

We also provide matching constructions that show $M^*$ and $\ell^*$ are optimal in strong ways (see Figure~\ref{fig:optimal}). An asymptotically good qLDPC code \cite{panteleev2021asymptotically,leverrier2022quantum} has $O(M^*)=O(\max(k,d))$ interactions of any length (see Theorem 1.3 of~\cite{dai2024locality}). Since a stabilizer code can also be trivially regarded as a subsystem code, this shows that our bounds are tight in terms of interaction count. For optimality in the interaction length, we exhibit subsystem codes embedded in $D$-dimensions where all interactions are of length at most $O(\ell^*)$ (Theorem~\ref{thm:constr-1}).

\begin{figure}
\centering
\begin{tikzpicture}[scale=5]
    \draw[->] (0,0) -- (0,1.05);
    \draw[->] (0,0) -- (1.05,0);
    \draw[] (0,0) -- (0,1) -- (1,1) -- (1,0);

    \foreach \i in {4,6,...,10} {
            \tikzmath{\hue = 20+8*\i; \xl = 1.0*\i / 10; \xlprev = 1.0*(\i-2)/10;}
            \draw[thick,blue!\hue] (0,\xl) -- (\xl,\xl) -- (\xl,0);
            \ifthenelse{\i = 10}{
                \node[color=blue!\hue] () at (\xl-0.08,\xl+0.05) {$M^*=n^{}$};
            }{
            \node[color=blue!\hue] () at (\xl-0.05,\xl+0.05) {$M^*=n^{0.\i}$};
            }
    }
    \node at (1.15,-0.0) () {$\log_n k$};
    \node at (0,1.1) () {$\log_n d$};
    \node at (-.1,0.0) () {$0$};
    \node at (0.0,-.05) () {$0$};
    \node at (-.1,0.5) () {$1/2$};
    \node at (0.5,-.05) () {$1/2$};
    \node at (-.1,1) () {$1$};
    \node at (1,-.05) () {$1$};
    \draw[fill=white] (0,0) -- (0,0.5) -- (0.5,0.5) -- (1,0) -- (0,0);
    \node[color=blue!20] at (0.3,0.2) {$M^*=0$};

    \node (bra) at (0.3,0.1) {\cite{bravyi2011subsystem}};
    \draw[->] (bra) to[out=0,in=240] (0.7,0.3);

    \node at (0.5,1.2) {$M^*$: Optimal Interaction Count};
\end{tikzpicture}
\begin{tikzpicture}[scale=5]
    \draw[->] (0,0) -- (0,1.05);
    \draw[->] (0,0) -- (1.05,0);
    \draw[] (0,0) -- (0,1) -- (1,1) -- (1,0) -- (0.5,0.5) -- (0,0.5);
    \draw[dotted] (0,0) -- (1,1);  

    \foreach \i in {1,...,5} {
            \tikzmath{\hue = 20+8*\i; \xl = 1.0*\i / 10; \xlprev = 1.0*(\i-1)/10;}
            \draw[thick,blue!\hue] (0,0.5 + \xl) -- (0.5+\xl,0.5+\xl) -- (1,2*\xl);
            \node[color=blue!\hue] () at (1.15, 2*\xl) {$\ell^*=n^{0.\i}$};
    }
    \node at (1.15,-0.0) () {$\log_n k$};
    \node at (0,1.1) () {$\log_n d$};
    \node at (-.1,0.0) () {$0$};
    \node at (0.0,-.05) () {$0$};
    \node at (-.1,0.5) () {$1/2$};
    \node at (0.5,-.05) () {$1/2$};
    \node at (-.1,1) () {$1$};
    \node at (1,-.05) () {$1$};
    \node[color=blue!20] at (0.3,0.2) {$\ell^*=1$};
    \node[rotate=0] at (0.75,0.5) {$\ell^*=\sqrt{\frac{kd}{n}}$};
    \node at (0.2,0.8) {$\ell^*=\frac{d}{\sqrt{n}}$};

    \node (bra) at (0.3,0.1) {\cite{bravyi2011subsystem}};
    \draw[->] (bra) to[out=0,in=240] (0.7,0.3);

    \node at (0.5,1.2) {$\ell^*$: Optimal Interaction Length};
\end{tikzpicture}
\caption{The (asymptotically) optimal interaction count and length for subsystem codes in 2D: A $[[n,k,d]]$ subsystem code need at least $\Omega(M^*)$ interactions of length $\Omega(\ell^*)$, where $M^*$ is plotted on the left and $\ell^*$ is plotted on the right. 
Above, we plot the contours of $k$ vs. $d$ tradeoffs for various values of the Interaction Count or Interaction Length.
Everywhere, big-$O$ is suppressed for clarity.}
\label{fig:main}
\end{figure}
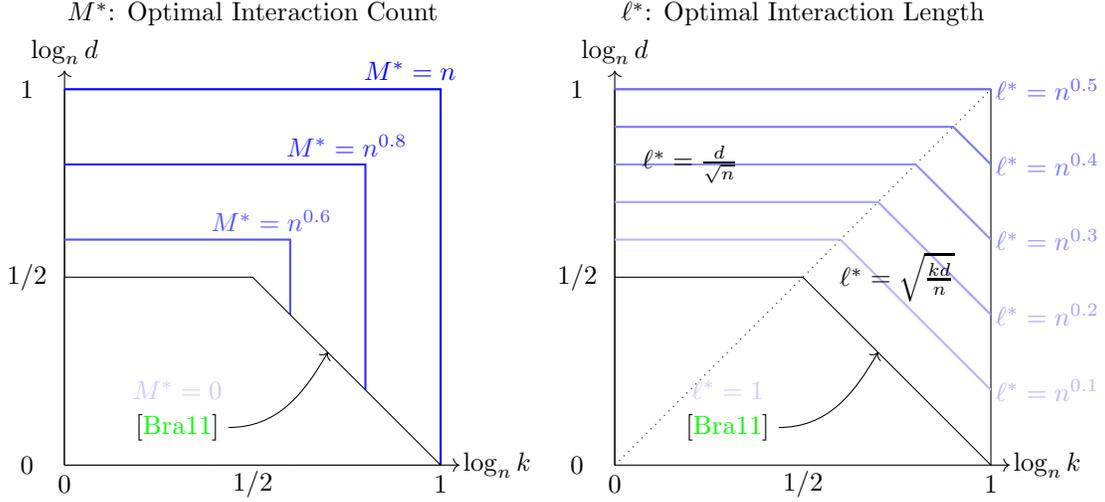

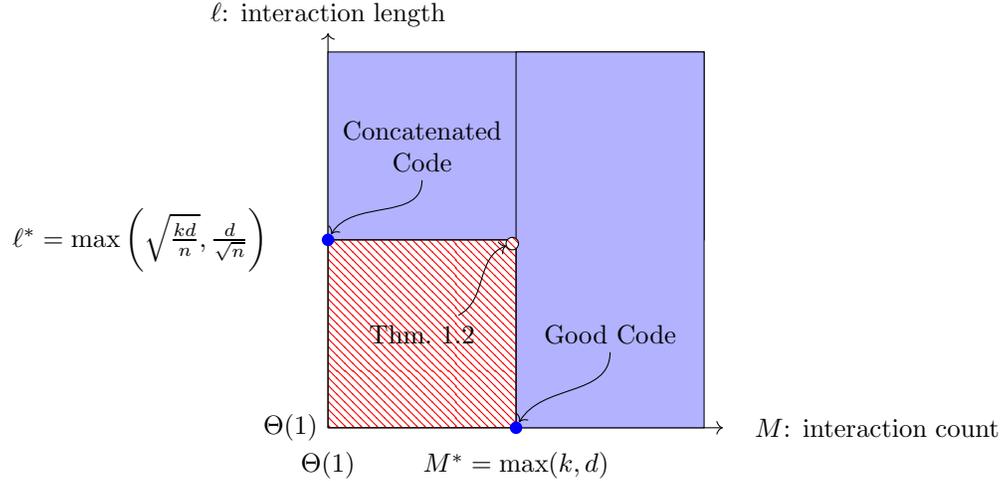
\begin{figure}
\centering
\hypersetup{linkcolor=black}

\begin{tikzpicture}[scale=5]
    \draw[->] (0,0) -- (0,1.05);
    \draw[->] (0,0) -- (1.05,0);
    \draw[fill=blue!30] (0,0.5) rectangle (1,1);
    \draw[fill=blue!30] (0.5,0.00) rectangle (1,1);
    \draw[pattern color=red,pattern=north west lines] (0,0) rectangle (0.5,0.5);
    \node at (1.46,-0.0) () {$M$: interaction count};
    \node at (0,1.1) () {$\ell$: interaction length};
    \node at (-.1,0.0) () {$\Theta(1)$};
    \node at (0.0,-.1) () {$\Theta(1)$};
    \node at (-0.5,0.5) () {$\ell^*=\max\left( \sqrt{\frac{kd}{n}}, \frac{d}{\sqrt{n}} \right)$};
    \node at (0.5,-.1) () {$M^*=\max(k,d)$};
    \node[preaction={fill,white},pattern color=red,pattern=north west lines,circle,scale=0.5, draw] (bad) at (0.49,0.49) {};
    \node[fill=blue,circle,scale=0.5] (good) at (0.5,0.00) {};
    \node[fill=blue,circle,scale=0.5] (concat) at (0,0.5) {};

    \node[align=center] (concat-label) at (0.25, 0.75) {Concatenated\\Code};
    \node[] (good-label) at (0.75, 0.25) {Good Code};
    \draw[->] (concat-label) to[out=270,in=60] (concat);
    \draw[->] (good-label) to[out=270,in=60] (good);

    \node[] (bad-label) at (0.25, 0.25) {Thm.~\ref{thm:main}};
    \draw[->] (bad-label) to[out=27,in=200] (bad);

    \node at (0.5,1.3) {\textbf{Interaction Count vs Length for Subsystem Codes}};
\end{tikzpicture}
\caption{Schematic diagram illustrating the optimality of our lower bounds for all $n,k,d$: A point $(M,L)$ represents that there is a code with $O(M)$ interactions of length $\omega(L)$. Blue shaded region is achievable, red lined region is unachievable. Our lower bound shows that $(M,\ell)$ with $M\le o(M^*)$ and $\ell\le o\left(\ell^*\right)$ is impossible, where $M^*$ and $\ell^*$ are the optimal interaction count and length, respectively, given by Theorem~\ref{thm:main}. There is a construction (good qLDPC code) with $O(M^*)$ interactions of any length, and another construction (concatenated local code, Theorem~\ref{thm:constr-1}) with zero interactions of length $\omega\left(\ell^*\right)$.}
\label{fig:optimal}
\end{figure}

\subsection{Generalizing \cite{dai2024locality} to $D$-dimensions}

We also show that the bounds in \cite{dai2024locality} can be generalized to $D$-dimensional embeddings and to commuting projector codes.

\begin{theorem}[Generalization of \cite{dai2024locality} to $D$-dimensions]
\label{thm:stab}
For any $D\ge 2$, there exist constants $c_0=c_0(D)>0$ and $c_1=c_1(D) > 0$ such that the following is true: Any $D$-dimensional embedding of a nontrivial $[[n,k,d]]$ commuting projector code with $kd^\frac{2}{D-1} \geq c_1 n$ or $d\ge c_1 n^{\frac{D-1}{D}}$ must have at least $M^*$ interactions of length $\ell^*$, where 
\begin{equation}
     M^* = c_0 \cdot \max(k,d), \qquad\text{and}\qquad \ell^* = c_0 \cdot \max\bigg(\frac{d}{n^\frac{D-1}{D}}, \bigg(\frac{kd^\frac{2}{D-1}}{n}\bigg)^\frac{D-1}{2D}\bigg).
\end{equation}
\end{theorem}

We now compare our works to prior works. First, we note that, when setting $D=2$, our bound matches the bounds in \cite{dai2024locality} up to the implied constant. 
For $D>2$ dimensions, the only prior bounds we are aware of are due to Baspin and Krishna \cite{baspin2022quantifying}.
Theirs match our bounds up to polylog factors when $d\ge \sqrt{kn}$ and when $k=\Theta(n)$, and we improve their bounds in the remaining parameter regimes. We also generalize their results; our results hold for all commuting projector codes, whereas theirs hold only for qLDPC codes.

Similar to Theorem~\ref{thm:main} and \cite{dai2024locality}, our $M^*$ and $L^*$ in Theorem~\ref{thm:stab} are optimal up to constant factors. Any asymptotically good qLDPC code \cite{panteleev2021asymptotically,leverrier2022quantum} is a commuting projector code with at most $O(M^*)=O(\max(k,d))$ interactions of any length. Further, we exhibit stabilizer codes embedded in $D$ dimensions, all of whose interactions are of length at most $O(\ell^*)$; see Theorem~\ref{thm:constr-2}.

\subsection{Organization of the Paper}

We divide the proof of Theorem~\ref{thm:main} into two parts: 
\newcommand\thmmaind{
For all $D\ge 2$, there exist constants $c_0 = c_0(D) >0$ and $c_1=c_1(D) > 0$ such that the following is true: Any $D$-dimensional embedding of a nontrivial $[[n,k,d]]$ subsystem or commuting projector code with  $d\ge c_1 n^{\frac{D-1}{D}}$  must have at least $c_0  d$ interactions of length at least $c_0 \frac{d}{n^\frac{D-1}{D}}$.
}
\begin{theorem}[Main Result -- Part 1]
\thmmaind
\label{thm:main-d}
\end{theorem}

\newcommand\thmmaink{
For all $D\ge 2$, there exist constants $c_0 = c_0(D)>0$ and $c_1=c_1(D) > 0$ such that the following is true: Any $D$-dimensional embedding of a $[[n,k,d]]$ subsystem code with $kd^\frac{1}{D-1} \geq c_1  n$ must have at least $c_0  k$ interactions of length at least $c_0 (\frac{kd^\frac{1}{D-1}}{n})^\frac{D-1}{D}$.
}
\begin{theorem}[Main Result -- Part 2]
\thmmaink
\label{thm:main-k}
\end{theorem}

Theorem~\ref{thm:main-d} implies Theorem~\ref{thm:main} in the regime where $d \geq k$, and Theorem~\ref{thm:main-k} implies Theorem~\ref{thm:main} in the regime when $d \leq k$. 
Note that Theorem~\ref{thm:main-d} also implies Theorem~\ref{thm:stab} when $d \geq \sqrt{kn}$. The remaining case of Theorem~\ref{thm:stab} is when $d \leq \sqrt{kn}$, which we prove in Theorem~\ref{thm:stab-1}. 

\newcommand\thmstab{
For all $D\ge 2$, there exist constants $c_0 = c_0(D)>0$ and $c_1=c_1(D) > 0$ such that the following is true: Any $D$-dimensional embedding of a $[[n,k,d]]$ commuting projector code with $kd^\frac{2}{D-1} \geq c_1n$ must have at least $c_0k$ interactions of length $c_0(\frac{kd^\frac{2}{D-1}}{n})^\frac{D-1}{2D}$.
}

\begin{theorem}[Generalization of \cite{dai2024locality} when $d \leq \sqrt{kn}$]
\thmstab
\label{thm:stab-1}
\end{theorem}

\section{Preliminaries}

\paragraph{Notation and Definitions.}

We use standard Landau notation $O(\cdot), \Omega(\cdot), \Theta(\cdot), o(\cdot), \omega(\cdot)$.
We also use the notations $\tilde O(\cdot),\ \tilde \Omega(\cdot)$, which are variants of $O(\cdot)$ and $\Omega(\cdot)$, respectively, that ignore logarithmic factors. For example, $f(n) =\tilde O(h(n))$ means that there exists an integer $k$ such that $f(n) = O(h(n) \log^k n)$.
For a set $S$, we write $S^{\le D}\defeq S\cup S^2\cup\cdots\cup S^D$.

In $\mathbb{R}^D$, \emph{distance} refers to Euclidean ($\ell_2$) distance unless otherwise specified. 
We sometimes also use the \emph{$\ell_\infty$-distance} of two points $(x,y),(x',y')\in\mathbb{R}^D$, which is $\max(|x-x'|,|y-y'|)$.
A \emph{grid tiling} is a division of $\mathbb{R}^d$ given by axis aligned hyperplanes equally spaced at a fixed distance $w$.
Throughout, a \emph{box} is always a set of the form $[a_1,b_1]\times\cdots\times[a_D,b_D]$. In particular, boxes contain their boundary and are axis-parallel.
A \emph{cube} is a box all of whose side lengths are equal: $b_1-a_1=b_2-a_2=\cdots=b_D-a_D$.

An \emph{embedded set} in $\mathbb{R}^D$ is a finite set $Q\subset\mathbb{R}^D$ with pairwise ($\ell_2$) distance at least 1. A function $f:\mathbb{R}^D\to\mathbb{N}$ is \emph{finitely supported} if $f(x)\neq 0$ for finitely many $x\in\mathbb{R}^D$. For a finitely supported function $f:\mathbb{R}^D\to\mathbb{N}$ and a region $R\subset\mathbb{R}^D$, define, by abuse of notation, $f(R)=\sum_{i\in R; f(i)\neq 0} f(i)$. 
We will be primarily concerned with the finitely supported function given by Definition~\ref{def:f}.

\subsection{Quantum codes}

We associate the pure states of a qubit with unit vectors in $\bbC^{2}$ and pure $n$-qubit states with unit vectors in $(\bbC^{2})^{\otimes n}$. Let $\mathcal{P}$ denote the (single-qubit) Pauli group, which consists of the Pauli matrices $\ssI, \ssX, \ssY, \ssZ$, and their scalar multiples by $\{\pm 1,\pm i\}$. The $n$-qubit Pauli group is $\mathcal{P}_n = \mathcal{P}^{\otimes n}$. Given $P \in \calP_n$, its \emph{weight} $|P|$ is the number of tensor components not equal to $\ssI$.

A \emph{quantum error-correcting code} $\calC$ is a subspace of $(\bbC^2)^{\otimes n}$. The parameter $n$ is called the \emph{(block) length} of the code. We define the \emph{dimension} $k$ of the code to be $k = \log_2(\dim \cal C)$.

\paragraph{Stabilizer Codes}
A \emph{stabilizer group} $\calS$ is an abelian subgroup of the $n$-qubit Pauli group $\mathcal{P}_n$ that does not contain $-\ssI$. A \emph{stabilizer code} $\mathcal{Q} = \mathcal{Q}(\mathcal{S}) \subseteq (\bbC^2)^{\otimes n}$ is defined to be the subspace of states left invariant under the action of the stabilizer group $\calS$, i.e.\ $ \mathcal{Q} = \{ \ket{\psi} : \; \ssS \ket{\psi} = \ket{\psi}, \;\forall \ssS \in \mathcal{S}\}$.
Being an abelian group, we can describe $\mathcal{S}$ by $n-k$ independent generators $\{\ssS_1,...,\ssS_{n-k}\}$, where $k$ is the \emph{dimension} of the code. The \emph{distance} $d$ is the minimum weight of an error $E \in \calP_n$ that maps a codeword in $\mathcal{Q}$ to another codeword. A quantum code $\mathcal{Q}$ with distance $d$ can correct up to $d-1$ qubit erasures. 

\paragraph{Subsystem Codes} A \emph{subsystem code} is a choice of decomposition of a stabilizer code $\calC$ into a tensor product $\calC = \calA \otimes \calB$, where $\calA \cong (\bbC^2)^{\otimes k}$ and $\calB \cong (\bbC^2)^{\otimes g}$ are the \emph{logical} and \emph{gauge} parts of $\cal C$, respectively. The dimension $k$ of a subsystem code is defined as the number of qubits encoded in its logical subsystem $\calA$. One can view a subsystem code as a stabilizer code that can encode $k+g$ logical qubits, but only $k$ of the logical qubits are actually used to protect information.

We can define a subsystem code by starting with a stabilizer code $\mathcal{S}$ given by $n-k-g$ independent stabilizer generators, with $k+g$ logical qubits associated with $k+g$ pairs of logical operators $\bar X_1,\bar Z_1,\dots,\bar X_{k+g}, \bar Z_{k+g}$. The first $k$ logical qubits are used to encode information, and the last $g$ logical qubits are called \emph{gauge qubits}. The \emph{gauge group} of the subsystem is the group $\mathcal{G}=\ab{\mathcal{S}, \bar X_{k+1},\bar Z_{k+1},\dots,\bar X_{k+g}, \bar Z_{k+g}}$. Given the gauge group $\calG$, the code's stabilizer group $\calS$ can be recovered as the center of $\cal G$, so a subsystem code is uniquely defined by its gauge group. Any stabilizer code can be equivalently regarded as a subsystem code whose gauge group is abelian, so stabilizer codes form a subset of subsystem codes.

For subsystem codes, we make a distinction between bare logical operations, which act trivially on the gauge qubits, and dressed logical operators, which may not. Formally, (non-trivial) bare logical operators are elements of $\sansserif{C}(\calG)\setminus\calG$, where $\sansserif{C}(\calG)$ denotes the centralizer of $\calG$, and (non-trivial) dressed logical operators are elements of $\sansserif{C}(\calS)\setminus\calG$. Note that for stabilizer codes there is no distinction. The distance $d$ of a subsystem code is defined as the minimum weight of a non-trivial dressed logical operator, i.e., $d=\min_{P\in \sansserif{C}(\calG)\setminus\calG}|P|$. We will sometime denote a subsystem code $\calC$ with $n$ physical qubits, $k$ logical qubits, distance $d$, and $g$ gauge qubits by $\calC = [[n,k,d,g]]$. 

\paragraph{Commuting Projector Codes}
A commuting projector code $\calC\subseteq (\bbC)^{\otimes n}$ is a subspace defined by a set of pairwise commuting projections $\{\Pi_1, \dots, \Pi_m\}$. The code $\calC$ is the subspace of states left invariant by all projections $\Pi_i$. Every stabilizer code is also a commuting projector code where the defining projections are of Pauli type, i.e., $\Pi_{i} = (\ssI+P_i)/2$, for some Pauli operator $P_i \in \calP_n$. For the purposes of establishing our locality bounds, the only properties we need of commuting projector codes is the fact that all the properties of correctable sets for stabilizer codes, i.e., those listed in Lemma~\ref{lem:correctable}, continue to hold without modification for commuting projector codes~\cite{haah2012logical}. Finally, we note that while stabilizer codes can be considered a subset of both subsystem and commuting projector codes, there is no direct relation between subsystem and commuting projector codes themselves.

\paragraph{Quantum codes in $D$ dimensions.}
Given a finite set $S$, an \emph{embedding} of $S$ into $\bbR^D$ is a map $\iota:S\rightarrow \bbR^D$ such that $\|\iota(s_i)-\iota(s_j)\| \ge 1$ for all distinct $s_i,s_j\in S$. The image $\iota(S)$ is then said to be an \emph{embedded set}. Throughout, the embedding map will usually be implicit. We identify the qubits $Q$ of a quantum code with a $D$-dimensional embedded set, which we continue to call $Q\subset \mathbb{R}^D$ by abuse of notation. By further abuse of notation, we refer to $Q\subset\mathbb{R}^D$ as the embedding of the qubits $Q$. When the set of qubits $Q$ is understood, given a subset $V\subset Q$, we write $\overline{V}\defeq Q\setminus V$ to denote the \emph{complement} of $V$ in $Q$.

\begin{definition}[Interactions]
Given an embedding $Q\subset \bbR^D$ of a quantum code $C$, either a commuting projector code or a subsystem code, \emph{interactions} of the code are defined with respect to a specific set of generators for that code. In the case that $C$ is a commuting projector code with defining projections $\{\Pi_1,\cdots,\Pi_m\}$, we say that a pair of qubits $q,p\in Q$ define an \emph{interaction} if $p$ and $q$ are both in the support of some projection $\Pi_i$. Similarly, if $C$ is a subsystem code with a set of generators $\{G_1,\cdots,G_m\}$ for its gauge group, we say that $p,q\in Q$ define an \emph{interaction} if $p$ and $q$ are both in the support of some gauge generator $G_i$. In both cases, the \emph{length} of an interaction $(p,q)$ is defined to be the $\ell_2$ distance between $p$ and $q$. 

Interactions are always defined with respect to a particular set of generators (either projector or gauge), but throughout we assume that the generator set is fixed (but otherwise arbitrary), and thus the set of interactions is fixed as well. Note that for subsystem codes, interactions are always defined with respect to \emph{gauge generators} and not \emph{stabilizer generators}. In particular, since each stabilizer generator is generally a product of multiple gauge generators, it is possible for a subsystem code to have local gauge generators, but non-local stabilizer generators. Indeed, it is only in such cases that a separation in the locality bounds for stabilizer and subsystem codes is possible.
\end{definition}
 
In the proofs of our results, we typically consider a fixed interaction length $\ell$. We then refer to an interaction as \emph{bad} if the length is at least $\ell$ and as \emph{good} if the length is less than $\ell$. Intuitively, good interactions are easier to deal with for our proof; they are effectively local, and can be treated in a similar to how local interactions are treated in the original proofs of the BPT and Bravyi bounds. To control the number of bad interactions, we introduce the following function that counts the number of bad interactions that a particular qubit participates in:

\begin{definition}[Interaction counter]\label{def:counter}
  Given a quantum code with a $D$-dimensional embedding $Q\subset \mathbb{R}^D$, let $f_{\ge\ell}:\mathbb{R}^D\to \mathbb{N}$ denote the \emph{interaction counting function}, where $f_{\ge\ell}(q)$ for $q\in Q$ equals the number of interactions of length at least $\ell$ that qubit $q$ participates in, and $f_{\ge\ell}(\cdot)=0$ outside of $Q$.
  \label{def:f}
\end{definition}

\paragraph{Correctable sets.}Like in previous works \cite{baspin2022connectivity,baspin2022quantifying,baspin2024improved, dai2024locality}, we analyze the limitations of quantum codes using correctable sets. Intuitively, a subset $U\subset Q$ of qubits is correctable if the code can correct the erasure of the qubits in $U$. We state the definition of a correctable set for completeness, though we only interface with the definition indirectly using Lemmas~\ref{lem:correctable},~\ref{lem:ab}, and~\ref{lem:abc}.

\begin{definition}[Correctable set]
Let $U \subset Q$ be a subset of qubits in a quantum code, and let $\overline{U} = Q\backslash U$. Let $\mathcal{D}[\overline{U}]$ and $\mathcal{D}[Q]$ denote the space of density operators associated with the sets of qubits $\overline{U}$ and $Q$ respectively.
  The set $U$ is \emph{correctable} if there exists a recovery channel $\mathcal{R}: \mathcal{D}[\overline{U}] \to \mathcal{D}[Q]$ such that for any code state $\rho$, we have $\mathcal{R}(\Tr_{U}(\rho)) = \rho$.
\end{definition}

For an embedded set of qubits $Q\subset\bbR^D$, we will say that a region $R\subset\bbR^D$ is correctable if the subset of all qubits contained in $R$ is a correctable subset. As an abuse of terminology, we will often refer to subsets of qubits as regions. For stabilizer and commuting projector codes, a region being correctable is equivalent to having no non-trivial logical operators supported on that region. For subsystem codes, a region $U \subset Q$ is correctable if and only if no non-trivial \emph{dressed} logical operators are supported on $U$. Note that $U$ being correctable also implies that no non-trivial \emph{bare} logical operators are supported on $U$, but the converse does not necessarily hold. This motivates the following definition.

\begin{definition}[Dressed-Cleanable]
If there are no non-trivial \emph{bare} logical operators supported on a region $U\subset Q$, then we say that $U$ is \emph{dressed-cleanable}~\cite{pastawski2015fault}.
\end{definition}

We use the following notions in Lemma~\ref{lem:correctable} to reason about correctable sets in subsystem codes and commuting projector codes.
In a quantum code with qubits $Q$, say sets $U_1,\dots,U_\ell\subset Q$ are \emph{decoupled} if there are no interactions between two distinct $U_i$'s.
For a set $U\subset Q$, let $\partial U = \partial_+ U \cup \partial_{-} U$ be the \emph{boundary} of $U$, where $\partial_+ U$ denotes the \emph{outer boundary} of $U$, the set of qubits outside $U$ that have an interaction with $U$, and $\partial_- U = \partial_+ \overline{U}$ is the \emph{inner boundary} of $U$. 
\begin{lemma}[\cite{bravyi2009no,haah2012logical,pastawski2015fault}]
  Let $Q$ be the qubits of a $[[n,k,d]]$ commuting projector or subsystem code $\calC$.
\label{lem:correctable}
\begin{enumerate}
   \item \label{lem:trivial} \textbf{\emph{Subset Closure:}} Let $U \subset Q$ be a correctable set. 
   Then any subset $W \subset U$ is correctable.
   \item \label{lem:dist} \textbf{\emph{Distance Property:}} Let $U \subset Q$ with $|U| < d$. Then $U$ is correctable.
   \item \label{lem:bptunion} \textbf{\emph{Union Lemma:}}
    Let $U_1,\dots,U_\ell$ be decoupled, and let each $U_i$ be correctable. If $\calC$ is a subsystem code, then $\bigcup_{i=1}^\ell U_i$ is dressed-cleanable. If $\calC$ is a commuting projector code, then $\bigcup_{i=1}^\ell U_i$ is correctable. 
    \item \label{lem:bptexpansion} \textbf{\emph{Expansion Lemma:}}
    Let $U,T\subset Q$ be correctable sets such that $T \supset \partial U$. Then $T \cup U$ is correctable. 
\end{enumerate}
\end{lemma}

A key point in Lemma~\ref{lem:correctable} is that the union lemma differs for commuting projector and subsystem codes. For subsystem codes, the union of decoupled and correctable sets is \emph{not necessarily correctable} -- only dressed-cleanable \cite{pastawski2015fault}. In general, being dressed-cleanable is weaker than being correctable. One of the major problems with generalizing Theorem~\ref{thm:stab} from commuting projector codes to subsystem codes is that the union lemma for subsystem codes only allows the conclusion that the union of correctable sets is dressed-cleanable. This version of the union lemma is too weak to adapt the original proof of Theorem~\ref{thm:stab} to subsystem codes. Instead, we take an alternative approach in proving Theorem~\ref{thm:main-d} which is based solely on the expansion lemma. 

The usefulness of reasoning about correctable sets is that the sizes of the correctable sets in a quantum code directly give bounds on the parameters.

\begin{lemma}[$AB$ Lemma -- Implicit in \cite{bravyi2011subsystem}, Section VIII]
    \label{lem:ab}
    Suppose that the qubits $Q$ of a $[[n,k,d]]$ subsystem code can be partitioned as $Q = A \sqcup B$. If $A$ is dressed-cleanable, then
    \[
        k \leq \abs{B}. 
    \]
\end{lemma}
\begin{lemma}[$ABC$ Lemma \cite{bravyi2010tradeoffs}]
    \label{lem:abc}
    Suppose that the qubits $Q$ of a $[[n,k,d]]$ commuting projector code can be partitioned as $Q = A \sqcup B \sqcup C$. If $A$ and $B$ are correctable, then 
    \[
        k \leq \abs{C}.
    \]
\end{lemma}

\subsection{Geometric Lemmas}

In this section, we give two lemmas about $D$-dimensional embeddings of sets. The first, Lemma~\ref{lem:packing}, allows us to generalize our results from lattice embeddings to arbitrary embeddings.
\begin{lemma}[Point Density]
\label{lem:packing}
Let $R\subseteq \mathbb{R}^D$ be a box with side lengths $L_1 \ge \cdots \ge L_D$. Suppose $Q\subseteq \bbR^D$ is an embedded set. Then 
\begin{align}
|R\cap Q| \le \frac{2^D}{\mathrm{vol}(B_D)}\prod_{i=1}^D(1+L_i)
\end{align}
where $B_D$ is the unit ball in $\mathbb{R}^D$.
\end{lemma}
\begin{proof}
By the definition of an embedding, $R\cap Q$ defines a $(1-\delta)$-packing\footnote{Recall that given a region $R\subseteq \bbR^D$, a set of points $\{p_i\}_{i=1}^P \subseteq R$ is an $\varepsilon$-packing in $R$ if $\|p_i-p_j\| > \varepsilon$ for all $i\neq j$. The packing number $P(R,\varepsilon)$ is then the cardinality of the largest $\varepsilon$-packing in $R$. In our case, a qubit embedding defines a $(1-\delta)$-packing for all $\delta > 0$. Note that the $\delta$ is necessary since packings are defined with strict inequality while embeddings are not.} in $R$ for any $\delta > 0$. The $\varepsilon$-packing number $P(R,\varepsilon)$ of any subset $R\subseteq \mathbb{R}^D$ satisfies
\begin{align}
P(R,\varepsilon) \le \left(\frac{2}{\varepsilon}\right)^D\frac{\mathrm{vol}(R+\frac{\varepsilon}{2}B_D)}{\mathrm{vol}\left(B_D\right)},
\end{align}
where $B_D$ is the unit ball in $\mathbb{R}^D$, and where the sum of regions is the Minkowski sum. The inequality above follows since the balls $B(p_i,\varepsilon/2)$ of radius $\varepsilon/2$ centered around the points $\{p_i\}_{i=1}^P$ of the packing are disjoint, and we have
\begin{align}
\bigsqcup_{i=1}^{P(R,\varepsilon)}B(p_i,\varepsilon/2) \subset R+\frac{\varepsilon}{2}B_D.
\end{align}
Taking the volumes of both sides gives the inequality. Our region of interest is a rectangle, and we have $\varepsilon = 1-\delta$. Then we have
\begin{align}
\mathrm{vol}\left(R+\frac{1-\delta}{2}B_D\right) \le \mathrm{vol}\left(R+\frac{1}{2}B_D\right) \le \prod_{i=1}^D\left(1+L_i\right).
\end{align}
Combining everything, we get
\begin{align}
|R\cap Q| &\le P(R,1-\delta)\\
&\le (1-\delta)^{-D}\frac{2^D}{\mathrm{vol}(B_D)}\mathrm{vol}\left(R+\frac{1-\delta}{2}B_D\right)\\ &\le (1-\delta)^{-D}\frac{2^D}{\mathrm{vol}(B_D)}\prod_{i=1}^D(1+L_i).
\end{align}
Since the inequality holds for all $\delta > 0$, the result follows.
\end{proof}

The second lemma, Lemma~\ref{lem:tiling}, utilizes the probabilistic method to generate a grid tiling that allows us to maintain a convenient distribution of the qubits and bad interactions in our embeddings (see Figure~\ref{fig:tiling}). 
\begin{lemma}[Tiling Lemma]
\label{lem:tiling}
Let $X,Y \subseteq \mathbb{R}^D$ be two multi-sets. Let $w$ and $\ell$ be positive integers with $w \ge 4\ell$. There exists a tiling of $\mathbb{R}^D$ using hypercubes of side length $w$ such that:
\begin{enumerate}
    \item at most a $(4\ell D/w)^2$ fraction of points in $X$ are within $\ell^\infty$-distance $2\ell$ of a codimension-$2$ face of some hypercube,
    \item at most a $8\ell D/w$ fraction of points in $Y$ are within $\ell^\infty$-distance $2\ell$ of a codimension-$1$ face of any hypercube.
\end{enumerate}
\end{lemma}
\begin{proof}Let $x \in \mathbb{R}^D$ be a point sampled from the uniform distribution on $[0,w]^D$. Consider the tiling of $\mathbb{R}^D$ using side length $w$ hypercubes with origin located at $x$. Say that a point $X$ is bad if it is within $\ell^\infty$-distance $2\ell$ of a codimension-$2$ face of a hypercube. Say that a point $Y$ is bad if it is within $\ell^\infty$-distance $2\ell$ of some codimension-$1$ face of a hypercube. The probability that a point of $X$ is bad is bounded above by $2D^2(2\ell/w)^2$. The probability of a point of $Y$ being bad is bounded above by $4\ell D/w$. The expected number of bad $X$ points is then at most $2(2\ell D/w)^2|X|$, and the expected number of bad $Y$ points is at most $(4\ell D/w)|Y|$. By Markov's inequality, the probability that $X$ has more than $(8\ell D/w)|X|$ bad points is less than $1/2$. Likewise, the probability that $Y$ has more than $4(2D\ell/w)^2|Y|$ bad points is less than $1/2$. It follows that there exists some choice of $x \in \mathbb{R}^D$ such that there are at most a $(4\ell D/w)^2$ fraction of bad points in $X$, and at most a fraction $8\ell D/w$ of bad points in $Y$.
\end{proof}

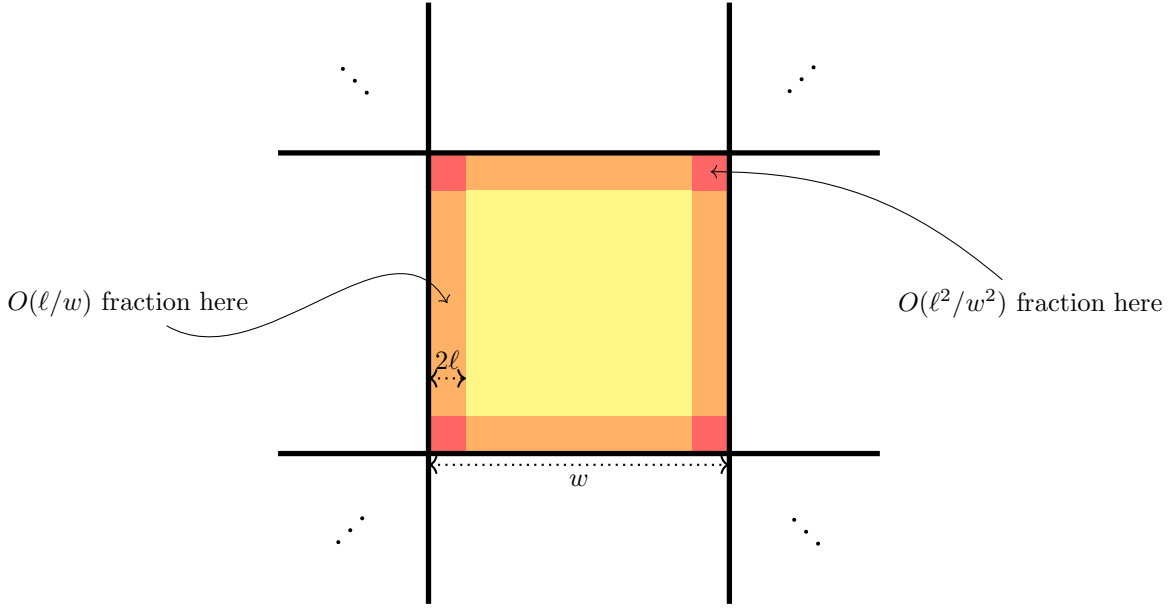
\begin{figure}
    \centering
    \begin{tikzpicture}[scale=1]
        \draw[draw=none,fill=orange!60] (-2.,-2.) rectangle (2,2);
        \draw[draw=none,fill=yellow!60] (-1.5, -1.5) rectangle (1.5, 1.5);
        \draw[draw=none,fill=red!60] (-2,-2) rectangle (-1.5,-1.5);
        \draw[draw=none,fill=red!60] (-2,2) rectangle (-1.5,1.5);
        \draw[draw=none,fill=red!60] (2,-2) rectangle (1.5,-1.5);
        \draw[draw=none,fill=red!60] (2,2) rectangle (1.5,1.5);
        \draw[line width=2pt] (-2,-4) -- (-2, 4);
        \draw[line width=2pt] (2,-4) -- (2, 4);
        \draw[line width=2pt] (-4,-2) -- (4, -2);
        \draw[line width=2pt] (-4,2) -- (4,2);
        \node[rotate=45] at (3,3) {\Large $\cdots$};
        \node[rotate=45] at (-3,-3) {\Large $\cdots$};
        \node[rotate=135] at (-3,3) {\Large $\cdots$};
        \node[rotate=135] at (3,-3) {\Large $\cdots$};

        \draw[<->, dotted, thick] (-2, -2.15) -- (2, -2.15) node [midway, below] {$w$};
        \draw[<->, dotted, thick] (-2., -1) -- (-1.5, -1) node [midway, above] {$2\ell$};

        \node[] (corner-label) at (6, 0) {$O(\ell^2/w^2)$ fraction here};
        \node[] (side-label) at (-6,0) {$O(\ell/w)$ fraction here};
    \draw[->] (corner-label) to[out=140,in=0] (1.75,1.75);
    \draw[->] (side-label) to[out=330,in=120] (-1.75,0);
    \end{tikzpicture}
\caption{Tiling Lemma: for fixed sets of points $X$ and $Y$ and a random width-$w$ grid tiling, we expect a $O(\ell^2/w^2)$ fraction of $X$ to be within a $O(\ell)$ of a grid codimension-2 face, and a $O(\ell/w)$ fraction of $Y$ to be within $O(\ell)$ of a codimension-1 face.}
\label{fig:tiling}
\end{figure}

\section{Proof of Theorem~\ref{thm:main-d}}
\label{sec:main-d}
We now prove Theorem~\ref{thm:main-d}, which covers the $d\ge k$ case of Theorem~\ref{thm:main}, our lower bound for subsystem codes. This also covers the $d\ge \sqrt{kn}$ case of Theorem~\ref{thm:stab}, our generalization of \cite{dai2024locality} to $D$-dimensions.

\begin{theorem*}[Theorem~\ref{thm:main-d}, restated]
    \thmmaind
\end{theorem*}

As mentioned after the statement of Lemma~\ref{lem:correctable}, the Union Lemma for subsystem codes is substantially weaker than the corresponding result for commuting projector codes. Without the ability to conclude that the union of correctable sets remains correctable, we cannot directly generalize the techniques previously employed in the proofs of the generalized BPT bound, which required alternating applications of the expansion and union lemmas~\cite{baspin2022connectivity,dai2024locality}. This poses a challenge for subsystem codes since we only obtain a dressed-cleanable set after the union lemma, and there is no straightforward way to continue with the expansion lemma, which requires correctable sets.

In view of these challenges, we take an alternative approach to the proof of Theorem~\ref{thm:main-d} by repeatedly --- and exclusively --- applying the expansion lemma to a carefully crafted subset of qubits in order to grow our correctable region. To do this, we grow our correctable region by sweeping across our set of qubits $Q\subseteq \bbR^D$ one dimension at a time, changing directions and moving into a new dimension whenever the expansion process in unable to continue in the previous dimensions. The brunt of the proof is showing that this process never gets stuck, and that we are eventually able to grow our correctable set without obstruction to encompass the entire set of qubits $Q$. In this way, we are able to bypass the usage of the union lemma altogether.

We point out that the usual way of doing this expansion \cite{bravyi2009no, bravyi2010tradeoffs, dai2024locality} --- starting with a $d^{1/D}\times\cdots\times d^{1/D}$ box and repeatedly adding the boundary --- does not work for general $D$-dimensional embeddings. For general embeddings, the qubits may not be constrained to an $O(n^{1/D})\times\cdots\times O(n^{1/D})$ box, so, at some step, the uncontrolled expansion boundaries could contain more than $d$ qubits, in which case we cannot apply the expansion lemma. For example, consider qubits embedded in 2-dimensions in a $n^{2/3}\times n^{2/3}$ square, with $\Omega(n)$ qubits distributed within $n^{1/3}$ of the box's boundary, and the remainder randomly distributed in the square. If we expand a rectangle from anywhere inside the square, applications of the expansion lemma fail when we approach the boundary. The easiest $D$-dimensional embeddings to realize are ones on a lattice structure contained in a $O(n^{1/D})\times\cdots\times O(n^{1/D})$ box; our general statement shows that more creative embeddings like the one above cannot save on locality.

The general expansion process is quite involved in $D$-dimensions, so we will first present an extended exposition of the $2$-dimensional case as an illustration of the basic idea of the proof. The full proof of Theorem~\ref{thm:main-d} will follow in the next subsection, and is similar in spirit to the $2$-dimensional case, although more involved in terms of notation and technique. The reader who is only interested in the general case may skip ahead to Section~\ref{sec:d-dim-proof}.

\subsection{Detailed Sketch of Theorem~\ref{thm:main-d} in $2$-dimensions}

In this section, we give a detailed sketch of the proof of Theorem~\ref{thm:main-d} for the case of $2$-dimensional embeddings. This is meant as a simplified illustration of the proof idea in the general $D$-dimensional case, which may be difficult to visualize.

Let $Q$ be a $2$-dimensional embedding of a subsystem code $\calC = [[n,k,d]]$ satisfying $d\ge 32\sqrt{n}$. Without loss of generality, we may assume that all of our qubits are contained in the interior of a large square $[\ell,A-\ell]\times [\ell,A-\ell]$, for some integer $A \in \bbN$. We remove a border of width $\ell$ so that we do not have to worry about edge cases later in our proof. Suppose for the sake of contradiction that there exists at most $d/4$ interactions of length at least $\ell = \frac{d}{32\sqrt{n}}$. Note that our choice of code parameters ensures that we have $1 \leq \ell \leq \sqrt{n}/32$. We will call any qubit participating in a length $\ge \ell$ interaction a bad qubit. Let $B$ be the set of all bad qubits. Then by assumption, $|B| \le d/2$. Given any region $R\subset \mathbb{R}^2$, we will say that $R$ is correctable if and only if $Q \cap R$ is correctable.

The basic idea of the proof is as follows. We wish to show that, given our assumption on the number of long interactions, a small correctable region can be iteratively grown without bound using the expansion lemma. Eventually the correctable region will encompass the entire set of qubits $Q$, which is a contradiction with our initial assumption that the code is non-trivial. We do this by first expanding along the $x_1$-direction, and then switching to the $x_2$-direction whenever we are unable to continue in the $x_1$-direction. The details for the two cases are given below.

\subsubsection{Expansion in the $x_1$-direction}

Let us imagine starting with a region of the form of a vertical strip $V[a_1] = [0,a_1]\times \bbR$, for some $a_1>0$. If $a_1<\ell$, then $V[a_1]$ contains no qubits by assumption, so it is vacuously correctable. Now, given an existing correctable region of the form $V[a_1]$, we wish to apply the expansion lemma to obtain a larger correctable set of the form $V[a_1+\ell]$. This will be possible provided that the number of qubits in the boundary of $V[a_1]$ is sufficiently small so that the boundary set itself is correctable. We will formalize this requirement by saying that a number $a_1 \in \bbR$ is ``good'' if the density of qubits around the line $x=a_1$ is low. Formally, $a_1 \in \bbR$ is a \emph{good} $x_1$-coordinate if the set $[a_i-\ell,a_1+\ell]\times \bbR$ contains at most $\ell\sqrt{n}$ qubits, and \emph{bad} otherwise. 

The boundary of $V[a_1]$ consists of all qubits within a distance $\ell$ of the line $x=a_1$, together with a subset of the bad qubits. Therefore the boundary is a subset of $B \cup [a_1-\ell,a_1+\ell]\times\bbR$. If $a_1$ is good, then by definition this subset contains at most 
\begin{align}
d/2 + \ell\sqrt{n} =  d/2 + d/32 < d
\end{align}
qubits, and is hence correctable. It follows by expansion and subset closure that if $V[a_1]$ is correctable and $a_1$ is good, then $V[a_1+\ell] \subset V[a_1] \cup \partial V[a_1]$ is also correctable. This gives us an easy way to grow the sets $V[a_1]$. However, this process cannot continue indefinitely, since there is no guarantee at each step that the new coordinate $a_1+\ell$ is good. When $a_1+\ell$ is a bad coordinate, our expansion process gets stuck. To get unstuck, we will fill in the stretch around the bad coordinate $a_1+\ell$ by expanding in the $x_2$-direction. After this gap containing bad $x_1$-coordinates has been filled, we can continue expanding in the $x_1$-direction until we reach our next obstacle.

\subsubsection{Stuck in the $x_1$-direction, expand along the $x_2$-direction}

Now we formalize the process of expanding in the $x_2$ direction. Given a bad coordinate $a_1$, we will define $\nxt(a_1)$ to be the ``next available'' good coordinate. More precisely, we define
\begin{align}
\nxt(a_1) = \inf G_{>a_1} + \gamma,
\end{align}
where $G_{>a_1}$ is the set of all good coordinates larger than $a_1$, and $\gamma>0$ is a sufficiently small value so that we actually have $\nxt(a_1) \in G_{>a_1}$.\footnote{The small constant $\gamma$ is necessary since the set of good coordinates is an open set, and as such does not contain its own infimum. A bit of thought reveals that it is sufficient to take $\gamma$ smaller than the minimum element of $\bigcup_{q\neq q'} \{|q_1-q_1'|, ||q_1-q_1'|-2\ell|\}$, where $q_1,q_1'$ denote the $x_1$-coordinates of qubits $q,q'\in Q$.} Given a set $V[a_1]$ where $a_1$ is good but $a_1+\ell$ is bad, let us define $V[a_1,a_2]$ to be
\begin{align}
V[a_1,a_2] = V[a_1]\cup ([a_1,\nxt(a_1+\ell)]\times[0,a_2]) = ([0,a_1]\times \bbR) \cup ([a_1,\nxt(a_1+\ell)]\times[0,a_2]).
\end{align}
We will call sets $V[a_1,a_2]$ defined this way \emph{legal} sets. Our goal is to show that given a correctable legal set $V[a_1,a_2]$, we can always apply expansion in the $x_2$-direction to obtain a new correctable legal set $V[a_1,a_2+\ell]$.

The key observation now is that the additional block $[a_1,\nxt(a_1+\ell)]\times[0,a_2]$ in $V[a_1,a_2]$ must be thin in the $x_1$-direction. Indeed, the entire interval $[a_1+\ell,\nxt(a_1+\ell)-\gamma]$ consists of bad coordinates. Any strip of width $2\ell$ centered around a bad coordinate contains at least $\ell\sqrt{n}$ qubits, and since we only have a total of $n$ qubits, the set $[a_1+\ell,\nxt(a_1+\ell)-\gamma]\times \bbR$ can fit at most $\floor{\sqrt{n}/\ell}$ such strips inside. It follows that we have
\begin{align}
\nxt(a_1+\ell) - \gamma - a_1 -\ell < 2\ell\left\lfloor\frac{\sqrt{n}}{\ell}\right\rfloor + 2\ell,
\end{align}
which implies
\begin{align}
\nxt(a_1+\ell) - a_1 < 2\sqrt{n} + 4\ell.
\end{align}
Now, consider the boundary of a correctable legal set $V[a_1,a_2]$. This will be a subset of 
\begin{align}
B \cup \underbrace{[a_1-\ell,a_1+\ell]\times \bbR}_{S_1} \cup \underbrace{[\nxt(a_1+\ell)-\ell,\nxt(a_1+\ell)+\ell]\times \bbR}_{S_2}  \cup \underbrace{[a_1,\nxt(a_1+\ell)]\times [a_2-\ell,a_2+\ell]}_{S_3}.
\end{align}
Since $a_1$ and $\nxt(a_1+\ell)$ are good coordinates, we have
\begin{align}
|Q\cap S_1| + |Q\cap S_2| \le 2\ell\sqrt{n}
\end{align}
by assumption. By Lemma~\ref{lem:packing}, the number of qubits in the subset $S_3$ is bounded above by
\begin{align}
|Q\cap S_3| \le \frac{4}{\pi}(2\ell+1)(2\sqrt{n}+4\ell+1) \le 9\ell\sqrt{n}.
\end{align}
Therefore the boundary of $V[a_1,a_2]$ contains at most
\begin{align}
|B| + |Q\cap S_1| + |Q\cap S_2| + |Q\cap S_3| \le \frac{d}{2} + 11\ell\sqrt{n} = \frac{27}{32}d < d
\end{align}
qubits, and is hence correctable. It follows by expansion and subset closure that if $V[a_1,a_2]$ is a correctable legal set, then $V[a_1,a_2+\ell]$ is again a correctable legal set. We can therefore continue expanding in the $x_2$-direction until we reach $V[a_1,A] = V[\nxt(a_1+\ell)]$. This is a vertical strip with a good boundary coordinate $\nxt(a_1+\ell)$. We can therefore return to the previous case and proceed to expand along the $x_1$-direction again. This process can now continue indefinitely, alternating between the $x_1$ and $x_2$-directions whenever we get stuck again. In this way, starting from an initial correctable set, say the vacuously correctable region $V[\ell/2]$, we can iteratively grow our correctable region without bound to encompass the entire set of qubits $Q$. Thus we arrive at our desired contradiction.

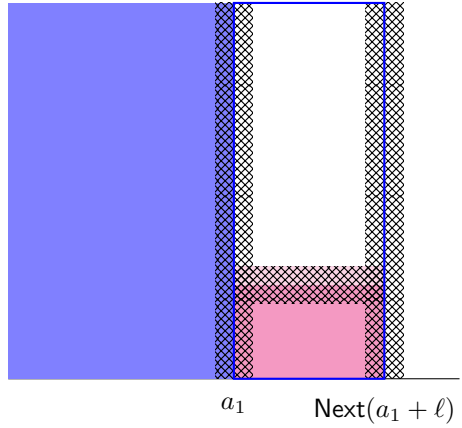
\begin{figure}
    \centering
    \begin{tikzpicture}[xscale=-1,yscale=-1]
        \draw (0,5)--(6,5);

        \draw [fill=blue!50,draw=none] (3,0) rectangle (6,5); 
      \draw [fill=magenta!20, draw=none] (1,3.5) rectangle (3,4);
        \draw[fill=magenta!50, draw=none](1,3.75) rectangle (3,5);
        \draw [pattern=crosshatch,pattern color=black, draw=none] (1,3.5) rectangle (3,4);

        \draw[fill=none,pattern=crosshatch, pattern color=black,draw=none](0.75,0) rectangle (1.25,5);

        \draw[fill=none,pattern=crosshatch, pattern color=black,draw=none](2.75,0) rectangle (3.25,5);

        \draw[draw=blue,thick] (1,0) rectangle (3,5);

        \node[label={-90:$\nxt(a_1+\ell)$}] at (1,5) {};
        \node[label={-90:$a_1$}] at (3,5) {};
    \end{tikzpicture}
    \caption{Sketch of the expansion process in 2 dimensions. The blue region is $V[a_1]$ and the pink region is $[a_1,\nxt(a_1+\ell)]\times [0,a_2]$. The crosshatched region (labeled $F$ in the proof of Theorem~\ref{thm:main-d}) contains the boundary of $V[a_1,a_2]$.}
    \label{fig:blue_boundary}
\end{figure}

\subsection{Proof of Theorem~\ref{thm:main-d} in $D$-dimensions}\label{sec:d-dim-proof}

For notational convenience, we grow the correctable region from low coordinates to high coordinates for this proof. For a set $S$, we write $S^{\le D}=S\cup S^2\cup\cdots\cup S^D$.

\begin{proof}
With hindsight, set $c_1=2(6^DD/\vol(B_D))^{D/(D-1)}$. Suppose $n,k,d$ are code parameters with $d\ge c_1 n^{\frac{D-1}{D}}$. With hindsight, choose 
\begin{align}
\ell = \frac{\vol(B_D)}{6^DDn^{(D-1)/D}} d,
\end{align}
so that $1<\ell< \frac{n^{1/D}}{4^DD}$. Let $Q$ be a $D$-dimensional embedding of a $[[n,k,d]]$ subsystem code $\calC$ with $d\ge k$, and suppose there are at most $d/4$ interactions of length at least $\ell$.

Let $B$ denote the set of qubits participating in long range interactions, so that $|B|\le d/2$.
Without loss of generality, we may assume all the qubits are in the box $[\ell,A-\ell]^D$ for some large integer $A$.
Choose a sufficiently small constant $\gamma>0$.
With hindsight, $\gamma$ less than the minimum nonzero element of $\bigcup_{i\in[D], q\neq q'} \{|q_i-q_i'|, ||q_i-q_i'|-2\ell|\}$ suffices.

For $i=1,\dots,D-1$, call real number $x_i\in\mathbb{R}$ \emph{$i$-good} if $\mathbb{R}^{i-1}\times[x_i-\ell,x_i+\ell]\times\mathbb{R}^{D-i}$ has at most $\ell n^{(D-1)/D}$ qubits and \emph{$i$-bad} otherwise.
By convention, we will consider every real number to be $i$-good when $i=D$.
An $i$-bad number represents an $i$-th-coordinate-value where we would ``get stuck'' and need to start expanding in a different direction.

For $i=1,\dots,D-1$ and an $i$-bad real number $x$, let $\nxt_i(x)$ be the maximum value such that all values in the interval $[x,\nxt_i(x)-\gamma]$ are $i$-bad ($\nxt_i(x)$ is undefined if $x$ is $i$-good).
Note, as $\gamma$ is sufficiently small, that $\nxt_i(x)$ is always $i$-good when it is defined.

For $a_1,\dots,a_i\in\mathbb{R}$, with $i \le D$, let
\begin{align}
  V[a_1,\dots,a_i]
    &= [0,a_1]\times\mathbb{R}^{D-1} \\
    &\ \quad\cup [a_1,\nxt_1(a_1+\ell)]\times[0,a_2]\times \mathbb{R}^{D-2} \\
    &\ \quad\cup [a_1,\nxt_1(a_1+\ell)]\times[a_2,\nxt_2(a_2+\ell)]\times [0,a_3]\times \mathbb{R}^{D-3} \\
    &\ \quad\qquad\qquad\qquad\qquad \vdots\qquad\qquad \vdots \qquad\qquad \vdots \\
    &\ \quad\cup [a_1,\nxt_1(a_1+\ell)]\times \cdots\times[a_{i-1},\nxt_{i-1}(a_{i-1}+\ell)]\times [0,a_i]\times \mathbb{R}^{D-i}.
\end{align}
Call such a set $V[a_1,a_2,\dots,a_{i}]$ \emph{legal} if (i) $a_j$ is $j$-good for all $j\in\{1,\dots,i\}$ and (ii) $a_j+\ell$ is $j$-bad for $j\in\{1,\dots,i-1\}$. We grow a large correctable set of the above form using the following four properties:
\begin{enumerate}
  \item\label{enum:part1} (If possible, expand in $i$th dimension): \emph{If $i\le D$, the set $V[a_1,\dots,a_{i}]$  is correctable and legal, and $V[a_1,\dots,a_{i-1},a_{i}+\ell]$ is legal, then $V[a_1,\dots,a_{i-1},a_i+\ell]$ is correctable.}

    We apply the expansion lemma. Note that the boundary is a subset of
    \begin{align}
      F
      &\defeq B \cup \left([a_1-\ell,a_1+\ell]\times\mathbb{R}^{D-1}\right) \cup \left([\nxt_1(a_1+\ell)-\ell,\nxt_1(a_1+\ell)+\ell]\times\mathbb{R}^{D-1}\right) \\
    &\qquad\qquad\qquad\qquad\qquad\qquad\qquad \vdots\qquad\qquad \vdots \qquad\qquad \vdots \\
    &\cup \left(\mathbb{R}^{i-2}\times [a_{i-1}-\ell,a_{i-1}+\ell]\times \mathbb{R}^{D-i+1}\right)\cup \left(\mathbb{R}^{i-2}\times [\nxt_{i-1}(a_{i-1}+\ell)-\ell,\nxt_{i-1}(a_{i-1}+\ell)+\ell]\times\mathbb{R}^{D-i+1}\right) \\
    &\cup \left(\mathbb{R}^{i-1}\times [a_i-\ell,a_i+\ell]\times \mathbb{R}^{D-i}\right) .
      \label{}
    \end{align}
    Since $V[a_1,\dots,a_i]$ is legal, the values $a_j$ and $\nxt(a_j+\ell)$ are $j$-good for $j\in\{1,\dots,i\}$ and $j\in\{1,\dots,i-1\}$, respectively. By definition of $i$-good, each set in the above union, other than $B$, has at most $\ell n^{(D-1)/D}$ qubits. Thus, $F$ has at most $d/2 + (2D-1)\ell n^{1/D} < d$ qubits, so $F$ is correctable. It follows that $V[a_1,\dots,a_i]\cup F$ is correctable, and by Subset Closure, $V[a_1,\dots,a_{i-1},a_i+\ell]$ is correctable.

  \item\label{enum:part2} (Stuck in $i$-th dimension, start in $(i+1)$-th-dimension): \emph{If $V[a_1,a_2,\dots,a_{i}]$ is correctable and legal, and $V[a_1,\dots,a_{i-1},a_i+\ell]$ is not legal, then $V[a_1,\dots,a_{i},0]$ is correctable and legal.}

    The set $V[a_1,\dots,a_{i},0]$ is correctable by definition as $V[a_1,\dots,a_{i},0]=V[a_1,\dots,a_{i}]$. It is legal also by definition, as we assume $a_i$ is $i$-good but $a_i+\ell$ is $i$-bad, and 0 is trivially $(i+1)$-good. 

  \item\label{enum:part3} (Expand in $D$-th dimension): \emph{If $V[a_1,a_2,\dots,a_{D}]$ is correctable and legal, then $V[a_1,\dots,a_{D-1},a_{D}+\ell]$ is correctable and legal.}

    It is legal as every real number is $D$-good by definition.
    For correctable, we again use the expansion lemma.
    Following part~\ref{enum:part1}, the boundary is a subset of
    \begin{align}
      F
      &\defeq B \cup \left([a_1-\ell,a_1+\ell]\times\mathbb{R}^{D-1}\right) \cup \left([\nxt_1(a_1+\ell)-\ell,\nxt_1(a_1+\ell)+\ell]\times\mathbb{R}^{D-1}\right) \\
    &\qquad\qquad\qquad\qquad\qquad\qquad\qquad \vdots\qquad\qquad \vdots \qquad\qquad \vdots \\
    &\cup \left(\mathbb{R}^{D-2}\times [a_{D-1}-\ell,a_{D-1}+\ell]\times \mathbb{R}\right)\cup \left(\mathbb{R}^{D-2}\times [\nxt_{D-1}(a_{D-1}+\ell)-\ell,\nxt_{D-1}(a_{D-1}+\ell)+\ell]\times\mathbb{R}\right) \\
    &\cup [a_1,\nxt_1(a_1+\ell)]\times \cdots\times[a_{D-1},\nxt_{D-1}(a_{D-1}+\ell)]\times [a_D-\ell,a_D+\ell].
    \end{align}
    As in part~\ref{enum:part1}, we have $|B|\le d/2$, and all but the last set in the union above has size at most $\ell n^{(D-1)/D}$. We now bound the size of the last set in $F$. Since $[a_j+\ell,\nxt_j(a_j+\ell)-\gamma]$ is $j$-bad for all $j$, a counting argument yields \begin{align}\nxt_j(a_j+\ell)-(a_j+\ell)\le 2n^{1/D}+2\ell.\end{align}
    To see this, pack strips of width $2\ell$ in the $j$th dimension into $\mathbb{R}^{j-1}\times[a_j+\ell,\nxt_{j}(a_j+\ell)]\times\mathbb{R}^{D-j}$. Each strip has at least $\ell n^{1/D}$ qubits by definition of being $j$-bad. There are at most $n$ qubits, so there are at most $n/(\ell n^{(D-1)/D})$ packed strips, so the total width satisfies
    \begin{align}
    \nxt_j(a_j+\ell)-(a_j+\ell)\le \frac{2\ell n}{\ell n^{(D-1)/D}}+2\ell=2n^{1/D} + 2\ell.
    \end{align}
    We conclude $\nxt_j(a_j+\ell)-a_j \le 2n^{1/D}+3\ell$. Hence, the last box has at most
    \begin{align}
    (2n^{1/D}+3\ell+1)^{D-1} (2\ell+1) < \frac{6^D}{\vol(B_D)} \ell n^{(D-1)/D}
    \end{align}
    qubits inside by Lemma~\ref{lem:packing}. The total boundary thus has at most 
    \begin{align}
    \frac{d}{2} + (2D-2)\ell n^{(D-1)/D} + \frac{6^D}{\vol(B_D)}\ell n^{(D-1)/D}  < d
    \end{align}
    qubits. It follows that $F$ is correctable, so $V[a_1,\dots,a_D]\cup F$ is correctable, and by Subset Closure, $V[a_1,\dots,a_{D-1}, a_D+\ell]$ is correctable.

  \item \label{enum:part4} (Finish $(i+1)$-th-dimension, get unstuck in $i$-th dimension): \emph{If $V[a_1,\dots,a_{i+1}]$ is correctable and legal, and $A-\ell\le a_{i+1}<A$, then $V[a_1,\dots,a_{i-1},\nxt_i(a_{i}+\ell)]$ is correctable}.

    The set is correctable because $V[a_1,\dots,a_{i-1},\nxt_i(a_{i}+\ell)]=V[a_1,\dots,a_i,\infty]=V[a_1,\dots,a_{i+1}]$.
    It is legal because $\nxt_i(a_i)$ is not $j$-bad by definition of $\nxt_i$.
\end{enumerate}
We can repeatedly apply these properties to get that the set of all qubits is correctable by induction.
Here are the details.
Let $\vec t_1,\vec t_2,\dots,$ be the enumeration of the $A_{\text{tot}}=(A+1)+(A+1)^2+\cdots+(A+1)^D$ tuples in $\{0,1,2, \dots,A\}^{\le D}$, in lexicographical order. 
Define the \emph{lexicographical index} of a region $V[a_1,\dots,a_i]$ as the largest $\alpha$ such that $t_\alpha$ is lexicographically less than or equal to $(a_1,\dots,a_i)\in\mathbb{R}^{\le D}$.
We prove by induction that, for all $r\le A_{\text{tot}}$, there exists a correctable and legal set with lexicographical index at least $r$.
For the base case, $V[0]=\emptyset$ is clearly correctable and legal. 
For the induction step, suppose we have a correctable and legal set $V[a_1,\dots,a_i]$ with lexicographical index $r$.
The above items shows that we can find a region with strictly larger lexicographical index:
If $a_i\ge A$, either $i=1$, in which case we are done, or $i\ge 2$ and we apply item~\ref{enum:part4} --- the lexicographical index increases because $\nxt_i(a_i+\ell)-a_i\ge \ell\ge 1$.
  Otherwise, if $i=D$, we apply item~\ref{enum:part3}.
  Otherwise, if $V[a_1,\dots,a_i+\ell]$ is legal, we apply item~\ref{enum:part1}, and if not, we apply item~\ref{enum:part2}.
  This completes the induction.

  Since the entire set of qubits $Q$ is correctable, an application of the AB Lemma~\eqref{lem:ab} with $A=Q$ and $B=\emptyset$ implies that $k = 0$. This contradicts our assumption that the code is nontrivial.
  \end{proof}

\section{Proof of Theorem~\ref{thm:main-k}}
We now prove Theorem~\ref{thm:main-k}, which covers the $k\ge d$ case of Theorem~\ref{thm:main}, our lower bound for subsystem codes.

\label{sec:main-k}
\begin{theorem*}[Theorem~\ref{thm:main-k}, restated]
    \thmmaink
\end{theorem*}

First, we prove two lemmas that help us find large correctable sets in a $D$-dimensional embedding of a quantum code. The first is a generalization of the ``holographic principle'' for error correction in \cite{bravyi2011subsystem}, which shows that the area, rather than the volume, governs the size of correctable sets. Recall that $f_{\ge\ell}(V)$ counts the number of interactions involving qubits in $V$ with length at least $\ell$ (see Definition~\ref{def:counter}).

\begin{lemma}[Holographic Principle]
    \label{lem:correctable-cube}
    Suppose we have a $[[n,k,d]]$ quantum code (either commuting projector or subsystem) with an embedding $Q \subset \mathbb{R}^D$, and suppose $\ell\le \frac{1}{8\sqrt{D}}d^{1/D}$. Let $V \subset Q$ be the subset of qubits contained in a box with sides of length at most 
    \begin{align}
        w_0\defeq \left(\frac{\vol(B_D)}{2\cdot 4^{D+1}D}\cdot\frac{d}{\ell}\right)^{\frac{1}{D-1}}.
    \end{align} 
    If $f_{\ge\ell}(V) \leq d/10$, then $V$ is correctable. 
\end{lemma}
\begin{proof}
    By subset closure, it suffices to prove the claim for a cube. The proof is by induction on $w$, the side length of the cube. The base case is for $w \leq \left(\frac{\vol(B_D)}{2\cdot 4^{D}}d\right)^{1/D}$. Then $V$ contains 
    \begin{equation}
        \frac{2^D}{\vol(B_D)} \cdot \prod_{i=1}^D (1 + w) \leq \frac{4^D}{\vol(B_D)} w^D \leq \frac{d}{2} < d
    \end{equation}
    qubits by Lemma~\ref{lem:packing}, so it is correctable. 

    For the induction step, assume that $\alpha \geq (\frac{\vol(B_D)}{2\cdot 4^{D}}d)^{1/D}$ and that the claim is true for all cubes of side length at most $\alpha$. Now let $w\le w_0$ satisfy $\alpha \leq w \leq \alpha + 2\ell$, and suppose that the qubits in $V$ are contained in a cube $H$ of side length $w$. Note that $2\ell \le \alpha \le w$. Consider two cubes $H^\pm$ with the same center as $H$ and side lengths $w \pm 2\ell$.
    Let $U$ be the set of qubits contained in $H_-$, so that $U$ is correctable by the induction hypothesis. Now let $T$ be the subset of qubits that are either:
    \begin{enumerate}[(i)]
        \item contained in $H\setminus H^-$, 
        \item contained in $H^+ \setminus H$, 
        \item participate in a long interaction with qubits in $U$, 
        \item in $U$ and participate in a long-ranged interaction with qubits in $Q \setminus U$, so that $\partial U \subset T$.
    \end{enumerate} 
    
    We now upper bound the number of qubits of each type. 
    Note that $H\setminus H^-$ can be covered by thin rectangular slabs with one side of length $\ell$ and all other sides at most $w$. The number of slabs required is $2D$. By Lemma~\ref{lem:packing}, the number of type (i) qubits is at most
    \begin{align}
        2D \cdot \frac{2^D}{\vol(B_D)} \cdot (1+\ell)\prod_{i=1}^{D-1} (1+w) &\leq \frac{2^{D+1}D}{\vol(B_D)} (2\ell)(2w)^{D-1} \le \frac{2^{D+1}D}{\vol(B_D)}(2\ell)(2w_0)^{D-1} = \frac{d}{4}.
    \end{align}
    Similarly, $H^+ \setminus H$ can be covered by rectangular slabs with one side of length $\ell$ and all other sides at most $w+2\ell$. Again, $2D$ such slabs are sufficient, so by Lemma~\ref{lem:packing} the number of type (ii) qubits is at most 
    \begin{align}
        2D \cdot \frac{2^D}{\vol(B_D)} \cdot (1+\ell) \prod_{i=1}^{D-1} (1+w+2\ell) \leq \frac{2^{D+1}D}{\vol(B_D)} (2\ell) (w + 3\ell)^{D-1} \leq \frac{2^{D+1}D}{\vol(B_D)} (2\ell) (2w_0)^{D-1} \leq \frac{d}{4},
    \end{align}
    where we've used the fact that $w_0\ge 3\ell$ in the second to last inequality. Finally, the numbers of type (iii) and type (iv) qubits are both at most $f_{\geq \ell}(V) \leq d/10$ by assumption.
    
    Adding up the counts of each type of qubit, we see that
    \begin{equation}
        \abs{T} \leq \frac{d}{4} + \frac{d}{4} + \frac{d}{10} + \frac{d}{10} < d,
    \end{equation} so that $T$ is correctable by the distance property. The expansion lemma then implies that $T \cup U$ is correctable as well. Finally, $V\subseteq T\cup U$ is correctable due to subset closure. This completes the induction step.
\end{proof}

\begin{figure}
    \centering
    \begin{tikzpicture}
        \draw [fill=blue!50,draw=none] (-2,-2) rectangle (2,2); 
        
        \draw[fill=magenta!50, draw=black](-2.5,-2.5) rectangle (2,-2);
        \draw[fill=magenta!50, draw=black](-2.5,-2) rectangle (-2,2.5);
        \draw[fill=magenta!50, draw=black](-2,2) rectangle (2.5,2.5);
        \draw[fill=magenta!50, draw=black](2,2) rectangle (2.5,-2.5);

        \draw[<->, dotted] (-1.8, -2) -- (-1.8, 2) node [midway, right] {$w$};
        \draw[<->, dotted] (-1.8, 2) -- (-1.8, 2.5) node [midway, right] {$\ell$}; 
        \draw[<->, dotted] (-2.5, 1.8) -- (-2, 1.8) node [midway, below] {$\ell$};

        \draw (-2, -2.5) -- (-2, -2);
        \draw (2, -2) -- (2.5, -2);
        \draw (-2.5, 2) -- (-2, 2);
        \draw (2, 2) -- (2, 2.5);
    \end{tikzpicture}
    
    \caption{The covering of the outer boundary $H^+\setminus H$ (magenta) used in the proof of Lemma~\ref{lem:correctable-cube} for $D=2$. The boundary is covered by four rectangles, each of width $\ell$ and length $w+2\ell$. The rectangles overlap at the corners. The generalization to higher dimensions involve a rectangular slab for each face of $H$ (blue).}    
    \label{fig:holographic-partition-2d}
\end{figure}
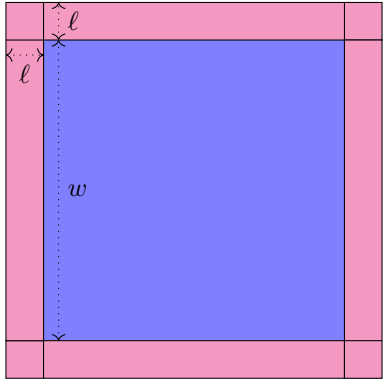

If we divide $\mathbb{R}^D$ into cubes using Lemma~\ref{lem:tiling}, most cubes will be ``good'' in that they contain $\ll d$ long-ranged interactions (see footnote $5$ of~\cite{dai2024locality}). Lemma~\ref{lem:correctable-cube} then says that all the good cubes are correctable. How do we handle the cubes with large numbers of bad qubits? In \cite{dai2024locality}, the solution was to further subdivide the bad cubes into sufficiently small rectangles, most of which then contains a sufficiently small number of bad qubits. The same strategy works in $D$-dimensions: 

\begin{lemma}[Subdivision]
Let $w,\ell$ and $d_1$ be positive real numbers. Let $f:\mathbb{R}^D\to \mathbb{N}$ be a finitely supported function. Let $R$ be a $h\times w^{D-1}$ box, with $h\ge 5\ell$ and $f(R)\ge d_1$. Then there exists a division of $R$ by hyperplanes orthogonal to $x_1$ into boxes $R_1,\dots,R_m$ such that:
    \begin{enumerate}
      \item Each box has dimensions $h_i\times w^{D-1}$, with $h_i\ge 5\ell$.
      \item Each $R_i$ satisfies either (i) $f(R_i)\le d_1$ or (ii) has $h_i\le 10\ell$. 
      \item The number of boxes $m$ is at most $\frac{2f(R)}{d_1}$.
    \end{enumerate}
    \label{lem:subdivision}
\end{lemma}
\begin{proof}
Let $V$ denote the support of $f$ in $R$, considered as a multiset (with each point $x$ having multiplicity $f(x)$). Let $V'$ be the projection of $V$ onto the $(x_1,x_2)$-plane, again as a multiset. Note that $V'$ is supported inside a $h\times w$ rectangle $R'$, and $f(R)=|V'|\ge d_1$. By the $2$-dimensional version of the subdivision lemma (see Lemma~4.4 in \cite{dai2024locality}), there exists lines orthogonal to $x_1$ that divide $R'$ into strips $R_1,\cdots,R_m$ such that:
\begin{enumerate}
\item Each strip $R_i$ has dimensions $h_i\times w$, with $h_i \ge 5\ell$.
\item Each strip $R_i$ satisfies either (i) $|V'\cap R_i| \le d_1 $, or (ii) $h_i \le 10\ell$. 
\item There are at most $2f(R)/d_1$ strips.
\end{enumerate}
The $D$-dimensional result follows by extending the dividing lines into hyperplanes orthogonal to $x_1$.
\end{proof}

Finally, we collect a few inequalities that we will frequently use in the proof of Theorems~\ref{thm:main-k} and~\ref{thm:stab-1}.

\begin{lemma}\label{lem:ineq}
Let 
\begin{align}
w_0 = \left(\frac{\vol(B_D)}{2\cdot 4^{D+1} D}\cdot\frac{d}{\ell}\right)^{\frac{1}{D-1}},\qquad\text{and}\qquad c = \frac{\vol(B_D)^\frac{1}{D}}{400\alpha D}
\end{align}
for some $\alpha \ge 1$. Suppose that $\ell$ satisfies $\ell \le cd^\frac{1}{D}$. Then we have
\begin{enumerate}
\item\label{item:ineq1} \hfill$\displaystyle\frac{2^D}{\vol(B_D)}(2w_0)^{D-1}\ell = \frac{d}{16D},$\hfill \ 

\item\label{item:ineq2} \hfill$\displaystyle w_0 \ge 100\alpha D\ell,$\hfill \ 

\item\label{item:ineq3} \hfill$\displaystyle
w_0 \ge \frac{1}{90\sqrt{D}}\left(\frac{d}{\ell}\right)^\frac{1}{D-1}.$\hfill\ 

\end{enumerate}
\end{lemma}
\begin{proof}
The first equality follows by explicit calculation:
\begin{align}
\frac{2^D}{\vol(B_D)}(2w_0)^{D-1}\ell = \frac{2^{2D-1}}{\vol(B_D)}\left(\frac{\vol(B_D)}{2\cdot 4^{D+1} D}\cdot\frac{d}{\ell}\right)\ell = \frac{d}{16D}.
\end{align}
For item~\ref{item:ineq2}, we have 
\begin{align}
w_0 &= \left(\frac{\vol(B_D)}{2\cdot 4^{D+1} D}\cdot\frac{d}{\ell}\right)^{\frac{1}{D-1}}\\ 
&= \left(\frac{\vol(B_D)}{2\cdot 4^{D+1} D}\frac{\ell^{D-1}}{c^D}\cdot\frac{c^Dd}{\ell^D}\right)^{\frac{1}{D-1}}\\ 
&\ge \left(\frac{\vol(B_D)}{8\cdot 4^Dc^{D}D}\right)^{\frac{1}{D-1}}\ell\\ 
&= \left(\frac{100^{D} \alpha^DD^{D-1}}{8}\right)^{\frac{1}{D-1}}\ell\\ 
&\ge 100\alpha D\ell,
\end{align}
where the inequality on the third line follows from the assumption that $\ell \le cd^\frac{1}{D}$, and the equality on the fourth line from the definition of $c$. Finally, item~\ref{item:ineq3} follows from
\begin{align}
w_0 = \left(\frac{\vol(B_D)}{2\cdot 4^{D+1} D}\cdot\frac{d}{\ell}\right)^{\frac{1}{D-1}} = \sqrt{D}\left(\frac{\vol(B_D)}{2\cdot 4^{D+1} D}\right)^\frac{1}{D-1} \frac{1}{\sqrt{D}}\left(\frac{d}{\ell}\right)^{\frac{1}{D-1}}.
\end{align}
Optimization of the function
\begin{align}
f(D) = \sqrt{D}\left(\frac{\vol(B_D)}{2\cdot 4^{D+1} D}\right)^\frac{1}{D-1}
\end{align}
shows that $f$ is strictly increasing on $[2,\infty)$, with $f(2) \ge 90^{-1}$.
\end{proof}

We are now ready to prove Theorem~\ref{thm:main-k}. 
\begin{proof}[Proof of Theorem~\ref{thm:main-k}]
    We give a proof by contradiction, where we first assume that we have an embedding of a subsystem code with $k$ logical qubits that has few long interactions, and then show that the code's dimension must actually be less than $k$. With hindsight, choose 
    \begin{align}
        c_0 = \frac{\vol(B_D)^\frac{1}{D}}{400D}.
    \end{align}
    Note that we have $c_0 \le 1/(200D) \le 1/400$ for $D\ge 2$.
    Choose $c_1 = (1/c_0)^\frac{D}{D-1}$. Suppose we have a $[[n,k,d]]$ subsystem  code with a $D$-dimensional embedding $Q \subset \mathbb{R}^D$ satisfying $kd^\frac{1}{D-1} \geq c_1 n$. Let 
    \begin{align}
    \ell = c_0\left(\frac{kd^\frac{1}{D-1}}{n}\right)^\frac{D-1}{D},
    \end{align}
    and note that we have $1 \leq \ell \leq c_0 d^\frac{1}{D} < \frac{1}{8\sqrt{D}}d^{\frac{1}{D}}$, where the first upper bound follows from $k \le n$, and the lower bound from $kd^\frac{1}{D-1} \geq c_1n$. 
    
    Now assume for the sake of contradiction that the embedding $Q \subset\mathbb{R}^D$ has at most $c_0 k$ interactions of length $\geq \ell$. Call an interaction \emph{long} if its length is at least $\ell$ and \emph{short} otherwise. Call a qubit $v \in Q$ \emph{bad} if it participates in a long interaction and \emph{good} otherwise. Then the function $f_{\geq \ell}(v)$ counts the number of long interactions that the qubit $v$ participates in. By assumption, there are at most $c_0k$ long interactions, so the total number of bad qubits is at most
    \begin{align}
    \sum_{v \in Q} f_{\geq \ell}(v) \leq 2c_0k.
    \end{align} 
   Now we construct a division of $\mathbb{R}^D$ into $\mathcal{A} \sqcup \mathcal{B} $ that outlines the partition of the qubits $Q = A \sqcup B$. Let 
   \begin{align}
   w_0 = \left(\frac{\vol(B_D)}{2\cdot 4^{D+1} D}\cdot\frac{d}{\ell}\right)^{\frac{1}{D-1}}
   \end{align}
   as in Lemma~\ref{lem:correctable-cube}. It follows from Lemma~\ref{lem:ineq}\eqref{item:ineq2} that $w_0\ge 100D\ell$. Apply Lemma~\ref{lem:tiling} with $Y = Q$ (and with $X$ arbitrary). This produces a tiling of $\mathbb{R}^D$ into cubes $\{S_m \}_{m \in \mathbb{Z}^D}$ of side length $w_0$, where at most $\frac{8D\ell}{w_0} n$ qubits of $Q$ are within $\ell_\infty$ distance $2\ell$ of some codimension-1 face of some cube.
   We call a cube $S_m$ \emph{good} if $f_{\geq \ell}(S_m) < d/10$ and \emph{bad} otherwise. Now apply Lemma~\ref{lem:subdivision}, with $d_1=d/10$, to decompose each bad cube $S_m$ into boxes $R_{m,1}, \cdots, R_{m,n_m}$. All boxes obtained in this way will also be called \emph{bad}. This process results in a division of $\mathbb{R}^D$ into good cubes and bad boxes. It follows form Lemma~\ref{lem:subdivision} (item 3) that total number of bad boxes is no more than 
    \begin{align}
        \sum_{m : S_m \text{bad}} \frac{2 f_{\geq \ell} (S_m)}{d/10} \leq \sum_{m} \frac{2f_{\geq \ell}(S_m)}{d/10} \leq \frac{20}{d} \sum_m f_{\geq \ell} (S_m) \leq \frac{40}{d} c_0k < \frac{k}{10d}.
    \end{align}
     Now we define the division $\mathcal{A} \sqcup \mathcal{B}$ as follows: 
    \begin{itemize}
        \item $\mathcal{B}$ is the set of all points within $\ell_\infty$ distance $2\ell$ of some codimension-1 face of either a good cube $S_m$ or a bad box $R_{m,i}$.
        \item $\mathcal{A}$ is the set of points not in $\mathcal{B}$.  
    \end{itemize}
    Note that we can perturb the tiling slightly in order to ensure that no qubits lie on the boundary of any subregion of $\mathcal{A}$ or $\mathcal{B}$. 
    
    Having constructed the division, we will now construct a corresponding partition of qubits $Q =  A \sqcup B$ such that $A$ is dressed-cleanable and $\abs{B} < k$. This will give us our desired contradiction from Lemma~\ref{lem:ab}. We define the partition $Q=A\sqcup B$ as follows: 
    \begin{itemize}
        \item $A$ is the set of all good qubits in region $\mathcal{A}$. 
        \item $B$ is the set of all remaining qubits. These are either good qubits in region $\mathcal{B}$ or bad qubits. 
    \end{itemize}
    Now we check that $A$ is dressed-cleanable and that $\abs{B} < k$.
    \begin{enumerate}
        \item {\bf $A$ is dressed-cleanable:} Let $\mathcal{A}_1', \mathcal{A}_2', \cdots$ be an arbitrary enumeration of the good cubes and bad boxes that divide $\mathbb{R}^D$. Let $A_i' \subset Q$ be the set of all qubits contained in region $\mathcal{A}_i' \subset \mathbb{R}^D$, and let $A_i = A_i'\cap A$ denote the subset of $A_i'$ contained in $A$. If $\mathcal{A}_i'$ is a good cube, then $A_i'$ is correctable by Lemma~\ref{lem:correctable-cube}. Otherwise, $\mathcal{A}_i'$ is a bad box, and either (i) $f_{\geq \ell}(\mathcal{A}_i') \leq d/10$, or (ii) $\mathcal{A}_i'$ has height at most $10\ell$. For case (i), $A_i'$ is again correctable by Lemma~\ref{lem:correctable-cube}. For case (ii), it follows by Lemma~\ref{lem:packing} and Lemma~\ref{lem:ineq}\eqref{item:ineq1} that $A_i'$ contains at most
        \begin{align}
            \frac{2^D}{\vol(B_D)} (1+w_0)^{D-1}(1+10\ell)  \le \frac{2^D}{\vol(B_D)} (2w_0)^{D-1}(11\ell) = \frac{11}{16D}d < d
        \end{align}
        qubits, so $A_i'$ is correctable by the the distance property. It follows by subset closure that each $A_i$ is correctable. Moreover, since the subsets $A_i$ and $A_j$ are separated by a distance of at least $\ell$, they are disjoint and decoupled for $i \neq j$. Applying the subsystem union lemma, we see that $A = \bigcup_i A_i$ is dressed-cleanable.
        
        \item {\bf $\abs{B} < k$}: The total number of bad qubits is at most $2c_0 k$ by assumption. Given a bad box $R_{m,i}$, it follows from Lemma~\ref{lem:packing} and Lemma~\ref{lem:ineq}\eqref{item:ineq1} that the number of qubits within $\ell_\infty$ distance $2\ell$ of a given codimension-1 face is at most
        \begin{align}
            \frac{2^D}{\vol(B_D)} (1+4\ell)(1+w_0+4\ell)^{D-1} &\leq \frac{2^D}{\vol(B_D)} (5\ell)(2w_0)^{D-1} \le  \frac{5}{16D} d.
        \end{align}
        Since each box has $2D$ total codimension-1 faces, and there are at most $k/(10d)$ bad boxes, the total number of qubits within $\ell_\infty$ distance $2\ell$ of any codimension-1 face of any bad box is at most 
        \begin{equation}
            \frac{k}{10d} \cdot 2D \cdot \frac{5}{16D}d = \frac{k}{16}.
        \end{equation}
        By our choice of grid tiling, the number of qubits within $\ell_\infty$ distance $2\ell$ of some codimension-1 face of some good cube $S_m$ is at most 
        \begin{align}
            \frac{8D\ell}{w_0}n &= 8D \ell \left(\frac{8\cdot4^{D}D}{\vol(B_D)} \cdot \frac{\ell}{d}\right)^{\frac{1}{D-1}}n \\ 
            &= 8^{\frac{D}{D-1}}\left(\frac{4D}{\vol(B_D)^\frac{1}{D}}\right)^\frac{D}{D-1}\ell^\frac{D}{D-1}\frac{n}{d^\frac{1}{D-1}}\\
            &= 8^{\frac{D}{D-1}}\left(\frac{1}{100c_0}\right)^\frac{D}{D-1}\left(c_0^\frac{D}{D-1}\cdot\frac{kd^\frac{1}{D-1}}{n}\right)\frac{n}{d^\frac{1}{D-1}}\\
            &= \frac{k}{(25/2)^{\frac{D}{D-1}}} < \frac{2}{25}k,
        \end{align}
        where we've substituted the definition of $c_0$ and $\ell$ in the third line. Summing the bad qubits and the qubits within $\ell_\infty$ distance $2\ell$ of a codimension-$1$ face of a good cube or bad rectangle, it follows that we have
        \begin{equation}
        \abs{B} \leq 2c_0 k + \frac{1}{16}k + \frac{2}{25}k \le \frac{1}{200}k + \frac{1}{16}k + \frac{2}{25}k < k.
        \end{equation}
    \end{enumerate}
    Since $A$ is dressed-cleanable, it follows by Lemma~\ref{lem:ab} that we must have $|B| \ge k$. This gives us our desired contradiction.
    
\end{proof}

\section{Proof of Theorem~\ref{thm:stab-1}}
\label{sec:cpc}

  \begin{figure}
  \begin{center}
      \begin{tikzpicture}[scale=0.45]
\draw[white,pattern=north west lines,pattern color=blue!40] (0,0) rectangle (15,15);
    \foreach \i in {0,5,...,15} {
      \draw [pattern=crosshatch,pattern color=red!100] (\i-0.2,0-0.2) rectangle (\i+0.2,15+0.2);
      \draw [pattern=crosshatch,pattern color=red!100] (0-0.2,\i-0.2) rectangle (15+0.2,\i+0.2);
    }
    \foreach \i in {0,5,...,15} {
      \foreach \j in {0,5,...,15} {
        \draw[fill=yellow!90] (\i-0.4,\j-0.4) rectangle (\i+0.4,\j+0.4);
      }
    }
    \draw [pattern=crosshatch,pattern color=red!100] (0-0.2,7-0.2) rectangle (5+0.2,7+0.2);
    \draw[fill=yellow!90] (0-0.4,7-0.4) rectangle (0+0.4,7+0.4);
    \draw[fill=yellow!90] (5-0.4,7-0.4) rectangle (5+0.4,7+0.4);

    \draw [pattern=crosshatch,pattern color=red!100] (10-0.2,13-0.2) rectangle (15+0.2,13+0.2);
    \draw[fill=yellow!90] (10-0.4,13-0.4) rectangle (10+0.4,13+0.4);
    \draw[fill=yellow!90] (15-0.4,13-0.4) rectangle (15+0.4,13+0.4);
    \node at (7.5,19) {\textbf{Case 1: $k\ge d$}};
    \node (a) at (1,17) {$\mathcal{A}$};
    \draw[->] (a) to[out=300,in=90] (2.5,12.5);

    \node (a) at (10,17) {Bad Rectangles};
    \draw[->] (a) to[out=330,in=90] (12,14);
    \draw[->] (a) to[out=330,in=90] (13,11.5);
    \draw[<->] (-0.7,10) -- (-0.7,15) node [midway, left] {$w$};
    \draw[<->] (-0.7,5-0.4) -- (-0.7,5+0.4) node [midway, left] {$\Theta(\ell)$};
  \end{tikzpicture} 
  \quad
      \begin{tikzpicture}[scale=0.45]
\draw[white,pattern=north west lines,pattern color=blue!40] (0,0) rectangle (15,15);
    \foreach \i in {0,5,...,15} {
      \draw [pattern=crosshatch,pattern color=red!30] (\i-0.4,-0.4) rectangle (\i+0.4,15.4);
      \draw [pattern=crosshatch,pattern color=red!30] (0-0.4,\i-0.4) rectangle (15+0.4,\i+0.4);
      \draw [pattern=crosshatch,pattern color=red!100] (\i-0.2,0-0.2) rectangle (\i+0.2,15+0.2);
      \draw [pattern=crosshatch,pattern color=red!100] (0-0.2,\i-0.2) rectangle (15+0.2,\i+0.2);
    }
    \foreach \i in {0,5,...,15} {
      \foreach \j in {0,5,...,15} {
        \draw[fill=yellow!90] (\i-0.4,\j-0.4) rectangle (\i+0.4,\j+0.4);
      }
    }
    \node[] at (7.5,19) {\textbf{Case 2: $d\ge k$}};
    \node (b) at (1,17) {$\mathcal{B}$};
    \node (b') at (3,17) {$\mathcal{B}'$};
    \node (c) at (8,17) {$\mathcal{C}$};

    \draw[->] (b) to[out=300,in=90] (2,15);
    \draw[->] (b') to[out=270,in=90] (3.5,15.4);
    \draw[->] (c) to[out=270,in=90] (10,15);
  \end{tikzpicture} 
  \end{center}
  \caption{The division of the plane into regions $\mathcal{A}$ (lined blue), $\mathcal{B}$ (red and pink crosshatch), and $\mathcal{C}$ (solid yellow) for the proof of Theorem~\ref{thm:stab-1}. The region $\mathcal{B}'$ (pink crosshatch) is also indicated in the figure on the right. These regions inform our qubit division $Q=A\sqcup B\sqcup C$. We use this division in different ways for the cases $k\ge d$ and $d\ge k$. When $k\ge d$ (left), we ignore $\mathcal{B}'$, and also subdivide any bad squares into bad rectangles. When $d\ge k$ (right), there are no bad squares, but we need to explicitly consider the region $\mathcal{B}'$.
  }
  \label{fig:abc}
  \end{figure}
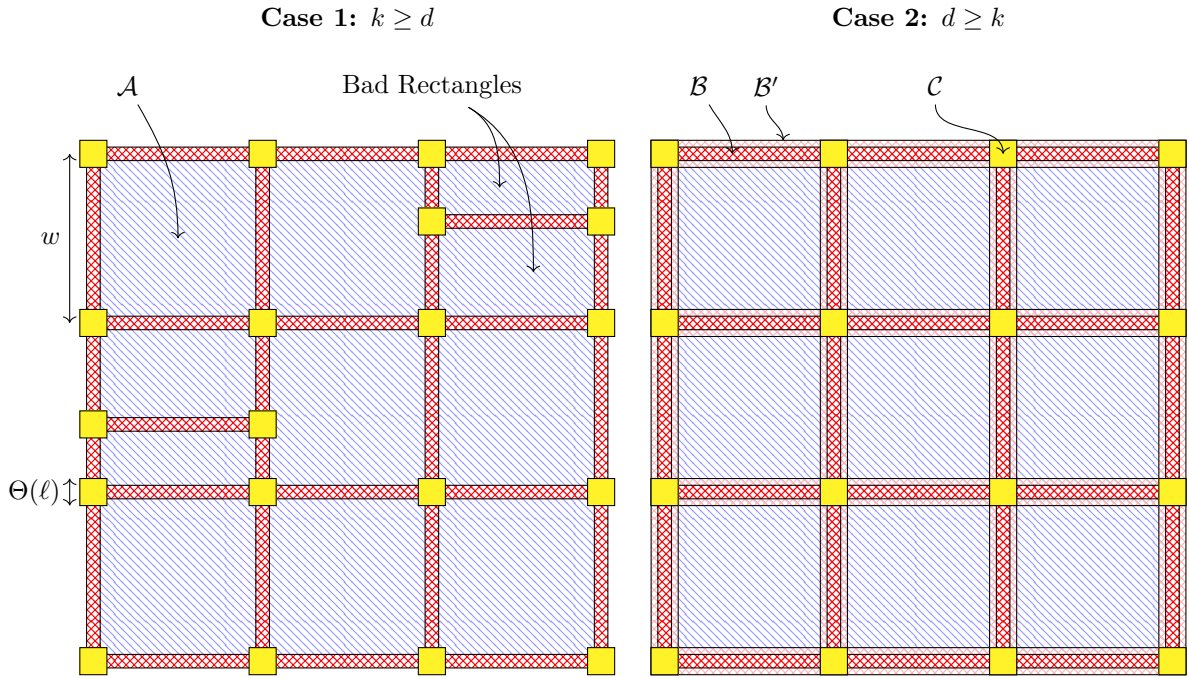

We now prove Theorem~\ref{thm:stab-1}, which, together with Theorem~\ref{thm:main-d}, yields Theorem~\ref{thm:stab}, our generalization of the main result of \cite{dai2024locality} to case of $D$-dimensional embeddings.

\begin{theorem*}[Theorem~\ref{thm:stab-1}, restated]
    \thmstab
\end{theorem*}

\begin{proof}
    For $d \geq \sqrt{kn}$, the result follows from Theorem~\ref{thm:main-d}, so it suffices to consider the case where $d \leq \sqrt{kn}$. With hindsight, choose
    \begin{align}
    c_0 = \frac{\vol(B_D)^\frac{1}{D}}{800D^2},
    \end{align}
    and let $c_1 = (1/c_0)^\frac{2D}{D-1}$. Note that $c_0 \le 1/(400D^2) \le 1/400$. Suppose we have a $[[n,k,d]]$ commuting projector code with a $D$-dimensional embedding $Q \subset \mathbb{R}^D$ satisfying $kd^\frac{2}{D-1} \geq c_1n$. Let 
    \begin{align}
    \ell = c_0\left(\frac{kd^\frac{2}{D-1}}{n}\right)^\frac{D-1}{2D}.
    \end{align}
    Note that we have $1 \leq \ell \leq c_0 d^\frac{1}{D} \le \frac{1}{8\sqrt{D}}d^{\frac{1}{D}}$, where the lower bound follows from $kd^\frac{2}{D-1} \geq c_1n$ and the first upper bound from the $k \le n$.
    
    Now assume for the sake of contradiction that the embedding $Q \subset\mathbb{R}^D$ has at most $c_0\max(k,d)$ interactions of length $\geq \ell$. Call an interaction \emph{long} if its length is at least $\ell$ and \emph{short} otherwise. Call a qubit $v \in Q$ \emph{bad} if it participates in a long interaction and \emph{good} otherwise. The function $f_{\geq \ell}(v)$ counts the number of long interactions that the qubit $v$ participates in. By assumption, the total number of long interactions is at most $c_0\max(k,d)$, so the total number of bad qubits it at most
    \begin{align}
    \sum_{v\in Q}f_{\ge \ell}(v) \le 2c_0\max(k,d).
    \end{align}
    We will construct a division of $\mathbb{R}^D$ into subsets $\mathcal{A} \sqcup \mathcal{B} \sqcup \mathcal{C}$ that will inform the partition of the qubits $Q = A \sqcup B \sqcup C$. Let 
    \begin{align}
    w_0 = \left(\frac{\vol(B_D)}{2\cdot 4^{D+1} D}\frac{d}{\ell}\right)^{\frac{1}{D-1}},\label{eq:bad_box2}
    \end{align}
    as in Lemma~\ref{lem:correctable-cube}. Note that it follows from Lemma~\ref{lem:ineq} that $w_0\ge 200D^2\ell$ and $w_0 \ge \frac{1}{90\sqrt{D}}(d/\ell)^\frac{1}{D-1}$. 
    
    Apply Lemma~\ref{lem:tiling} with $X = Q$ and with $Y$ as the multiset where each qubit $v$ appears with multiplicity $f_{\geq \ell} (v)$. This gives a partition of $\mathbb{R}^D$ into a set of cubes $\{S_m \}_{m \in \mathbb{Z}^D}$ of side length $w_0$. By construction, at most $\frac{16D^2\ell^2}{w_0^2}n$ qubits of $Q$ are within $\ell_\infty$ distance $2\ell$ of a codimension-2 face of some cube, and at most $\frac{8D\ell}{w_0}\cdot 2c_0\max(k,d)$ bad interactions involve a qubit within $\ell_\infty$ distance $2\ell$ of a codimension-1 face of some cube. We call a cube $S_m$ \emph{good} if $f_{\geq \ell}(S_m) < d/10$ and \emph{bad} otherwise. Now apply Lemma~\ref{lem:subdivision} to decompose each bad cube into boxes $R_{m,1}, \cdots, R_{m,n_m}$. All boxes obtained by subdividing a bad cube will also be called bad. This process results in a division of $\mathbb{R}^D$ into good cubes and bad boxes. By Lemma~\ref{lem:subdivision} (item 3), the total number of bad boxes is no more than 
    \begin{align}
        \sum_{m : S_m \text{bad}} \frac{2 f_{\geq \ell} (S_m)}{d/10} \leq \sum_{m} \frac{2f_{\geq \ell}(S_m)}{d/10} \leq \frac{20}{d} \sum_m f_{\geq \ell} (S_m) \leq \frac{40}{d} c_0 \max(k,d) < \frac{\max(k,d)}{10d}.
    \end{align}
     Now we define the division $\mathcal{A} \sqcup \mathcal{B} \sqcup \mathcal{C}$ as follows: 

    \begin{itemize}
        \item $\mathcal{C}$ is the set of all points within $\ell_\infty$ distance $2\ell$ of some codimension-2 face of a good cube $S_m$ or a bad box $R_{m,i}$. 
        \item $\mathcal{B}$ is the set of all points \emph{not} already in $\mathcal{C}$ and within $\ell_\infty$ distance $\ell$ of some codimension-1 face of a good cube $S_m$ or bad box $R_{m,i}$. 
        \item $\mathcal{B}' \subset \mathcal{B}$ is the set of all points \emph{not} already in $\mathcal{C}$ and within $\ell_\infty$ distance $2\ell$ of some codimension-1 face of a good cube $S_m$. 
        \item $\mathcal{A}$ is the set of points not in $\mathcal{B}$ or $\mathcal{C}$.  
    \end{itemize}
    Note that we can perturb the tiling slightly in order to ensure that no qubits lie on the boundary of any subregion of $\mathcal{A}$, $\mathcal{B}$, $\mathcal{B'}$, or $\mathcal{C}$. 

    Having defined the division of $\bbR^D$, we will now construct our partition of qubits $Q =  A \sqcup B \sqcup C$. A sketch of the high-level ideas are as follows: We aim to define our qubit partition with the goal of having $A,B$ be correctable, and $\abs{C} < k$. This will lead to the desired contradiction using Lemma~\ref{lem:abc}. There are two cases to consider, depending on whether $k \geq d$ or $k \leq d$. When $k \geq d$, we have $2c_0k \ll k$ bad qubits, which can then be directly placed in $C$ without affecting the requirement that $|C| < k$. When $k \leq d$, we have $2c_0 d \ll d$ bad qubits, and the set of all bad qubits is itself correctable. Our chosen division of $\bbR^D$ implies that very few bad qubits can interact with the qubits in $\mathcal{B}$, so the union lemma suggests that we can add almost all of the bad qubits to $\mathcal{B}$ while preserving its correctability. We now continue with our proof, divided into the two cases $k\ge d$ and $k \le d$.

    \paragraph{Case 1: $k \geq d$.} We define the partition of qubits $Q = A \sqcup B \sqcup C$ as follows: \
    \begin{itemize}
        \item $C$ is the set of qubits in region $\mathcal{C}$, along with all bad qubits. 
        \item $B$ is the set of all good qubits in region $\mathcal{B}$
        \item $A$ is the set of all good qubits in region $\mathcal{A}$. 
    \end{itemize}
    It's clear that this is indeed a partition of $Q$. Now we show that $A,B$ are correctable, and that $|C| < k$.
    \begin{itemize}
    \item {\bf $A$ is correctable}. Let $\mathcal{A}_1', \mathcal{A}_2', \dots$ be an arbitrary enumeration of the good cubes and bad boxes that divide $\mathbb{R}^D$. Let $A_i' \subset Q$ be the set of all qubits contained in the region $\mathcal{A}_i'$, and let $A_i = A_i'\cap A$ denote the subset of $A_i'$ contained in $A$. If $\mathcal{A}_i'$ is a good cube, then $A_i'$ is correctable by Lemma~\ref{lem:correctable-cube}. Otherwise, $\mathcal{A}_i'$ is a bad box, and either (i) $f_{\geq \ell}(\mathcal{A}_i') \leq d/10$, or (ii) $\mathcal{A}_i'$ has height at most $10\ell$. For case (i), $A_i'$ is again correctable by Lemma~\ref{lem:correctable-cube}. For case (ii), it follows by Lemma~\ref{lem:packing} and Lemma~\ref{lem:ineq}\eqref{item:ineq1} that $A_i'$ contains at most
    \begin{align}
        \frac{2^D}{\vol(B_D)} (1+w_0)^{D-1}(1+10\ell) &\leq \frac{2^D}{\vol(B_D)} (2w_0)^{D-1}(11\ell) \le \frac{11}{16D}d < d,
    \end{align}
    so $A_i'$ is correctable by the the distance property. It follows by subset closure that each $A_i$ is correctable. Moreover, since the subsets $A_i$ and $A_j$ are separated by a distance of at least $\ell$, they are disjoint and decoupled for $i \neq j$. Applying the union lemma, we see that $A = \bigcup_i A_i$ is correctable.
    
    \item {\bf $B$ is correctable.} We can divide the region $\calB$ into rectangular slabs of dimensions at most $2\ell \times w^{D-1}_0$. Each slab is essentially a thickening of a codimension-$1$ face of a good cube or bad box. Let $\mathcal{B}_1, \mathcal{B}_2, \cdots$ be an arbitrary enumeration of the slabs that decompose $\mathcal{B}$. Let $B_i$ be the set of qubits in $B$ that are contained in region $\mathcal{B}_i$. By Lemma~\ref{lem:packing} and Lemma~\ref{lem:ineq}\eqref{item:ineq1}, there are at most 
    \begin{align}
        \frac{2^D}{\vol(B_D)}(1+w_0)^{D-1}(1+2\ell) \leq \frac{2^D}{\vol(B_D)} (2w_0)^{D-1}3\ell  = \frac{3d}{16D} < d
        \label{eq:b-corr}
    \end{align}
    qubits in $B_i$, and so $B_i$ is correctable by the distance property. Since any two regions $\mathcal{B}_i$ and $\mathcal{B}_j$ are separated by distance at least $\ell\sqrt{2}$ and all the qubits in $B$ are good, it follows that $B_i$ and $B_j$ are disjoint and decoupled for $i \neq j$. Applying the union lemma, we conclude that $B = \bigcup_i B_i$ is correctable.
    
    \item {\bf $\abs{C} < k$.} Given a bad box $R_{m,i}$, the number of qubits within $\ell_\infty$ distance $2\ell$ of a given codimension-$2$ face is, by Lemma~\ref{lem:packing}, at most
    \begin{align}
        \frac{2^D}{\vol(B_D)} (1+4\ell)^2(1+w_0)^{D-2} &\leq \frac{2^D}{\vol(B_D)}(5\ell)^2(2w_0)^{D-2} \\
        &\leq \frac{2^D}{\vol(B_D)}(5\ell)^2\frac{(2w_0)^{D-1}}{100D^2\ell}\\
        &=\frac{2^D}{\vol(B_D)}\ell(2w_0)^{D-1}\frac{1}{4D^2}\\
        &\le \frac{1}{64D^3} d,
    \end{align}
    where we use the fact that $w_0 \ge 200D^2\ell$ for the second inequality, and Lemma~\ref{lem:ineq}\eqref{item:ineq1} in the last line. Each bad box has $\binom{D}{2}2^{D-(D-2)} \leq 2D^2$ codimension-2 faces, and the number of bad boxes $R_{m,i}$ is at most $k/(10d)$. Therefore the number of qubits within $\ell_\infty$ distance $2\ell$ of some codimension-$2$ face of some bad box is at most 
    \begin{align}
        \frac{k}{10d}\cdot 2D^2 \cdot \frac{1}{64D^3}d &= \frac{k}{320D}.
    \end{align} 
    By our choice of grid tiling, the number of qubits within $\ell_\infty$ distance $2\ell$ of some codimension-2 face of some good cube $S_m$ is at most
    \begin{align}
        \frac{16D^2\ell^2}{w^2_0}n &\le 16D^2\ell^2\cdot 90^2D \left(\frac{\ell}{d}\right)^\frac{2}{D-1}n\\
        &= 16D^390^2\ell^{\frac{2D}{D-1}} \frac{n}{d^\frac{2}{D-1}}\\ \label{eq:ineq1}
        &= 16D^390^2 c_0^\frac{2D}{D-1}k\\ 
        &\le 16D^390^2 c_0^2k\\
        &\le 16D^390^2 \cdot \frac{1}{(400D^2)^2}k\\ 
        &= \frac{81}{100D}k,
    \end{align}
    where we use $w_0 \ge \frac{1}{90\sqrt{D}}(d/\ell)^\frac{1}{D-1}$ in the first line, the definition of $\ell$ on the third, and $c_0 \le 1/(400D^2)$ on the fifth. Finally, the total number of bad qubits is at most $2c_0 \max(k,d)=2c_0k$ by assumption. Putting everything together, we find 
    \begin{equation}
        \abs{C} \leq 2c_0k + \frac{k}{32} + \frac{81}{100}k \leq \frac{k}{200D^2} + \frac{k}{64D} + \frac{81}{100D}k < \frac{1}{2}k.
    \end{equation}
    \end{itemize}
    This gives us a contradiction with the fact that $A$ and $B$ are correctable through Lemma~\ref{lem:abc}. 

\paragraph{Case 2: $d \geq k$.} From equation~\eqref{eq:bad_box2}, the total number of bad boxes is at most $\max(k,d)/(10d) = 1/10$, which is less than $1$. It follows that there are no bad boxes in this case, only good cubes. We define the partition of qubits $Q = A \sqcup B \sqcup C$ in this case as follows: 
\begin{itemize}
    \item $C$ consists of the set of qubits in region $\mathcal{C}$, together with all qubits participating in a long interaction with a qubit in $\mathcal{B}'$ (including the bad qubits in $\mathcal{B}')$. 
    \item $B$ is the set of good qubits in region $\mathcal{B}$, together with the bad qubits not in $C$.
    \item $A$ is the set of good qubits in region $\mathcal{A}$. 
\end{itemize}
It's clear that this is a partition of the qubits in $Q$. Now we check that $A,B$ are correctable and $\abs{C} < k$.

\begin{itemize}
    \item {\bf $A$ is correctable.} Every cube $S_m$ is good, so the set of qubits in $S_m$ is correctable by Lemma~\ref{lem:correctable-cube}. Let $A_m$ be the set of qubits in $A$ contained in $S_m$. By subset closure, $A_m$ is also correctable. All qubits in $A_m$ are good by definition. Moreover, for $m' \neq m$, the qubits in $A_m$ and $A_{m'}$ are separated by distance at least $2\ell$, so $A_{m}$ and $A_{m'}$ are disjoint and decoupled. By the union lemma, $A = \bigcup_m A_m$ is correctable.
    
    \item {\bf $B$ is correctable.} We can divide the region $\calB$ into rectangular slabs of dimensions $2\ell \times (w_0-2\ell)^{D-1}$. Each slab is essentially a thickening of a codimension-$1$ face of a good cube. Arbitrarily enumerate these slabs $\mathcal{B}_1, \mathcal{B}_2, \cdots$, and let $B_i$ denote the subset of qubits contained in region $\mathcal{B}_i$. It follows from Lemma~$\ref{lem:packing}$ and Lemma~\ref{lem:ineq}\eqref{item:ineq1} that each $B_i$ contains at most
    \begin{align}
        \frac{2^D}{\vol(B_D)} (1+2\ell)(1+w_0-2\ell)^{D-1} \leq \frac{2^D}{\vol(B_D)} (3\ell)(2w_0)^{D-1} \le \frac{3d}{16D^2} < d 
    \end{align}
    qubits. Therefore $B_i$ is correctable by the distance property. Moreover, $B_1, B_2, \cdots$ contain only good qubits by construction and the regions $\mathcal{B}_1, \mathcal{B}_2, \cdots$ are pairwise separated by distance at least $\ell\sqrt{2}$. It follows that $B_1, B_2, \cdots$ are pairwise disjoint and decoupled. 
    
    Now let $B_0 \subset B$ be the set of bad qubits that are not contained in $C$. By assumption we have at most $2c_0 d$ bad qubits. Since $c_0 < \frac{1}{500}$, we have $\abs{B_0} < d/250 < d$, so $B_0$ is correctable by the distance property. By construction $B_0$ lies outside of the region $\calB'$ and does not contain any qubits participating in long-ranged interactions with the qubits of $\mathcal{B}'$. Since $B_0$ is separated from $B_1,B_2,\cdots$ by distance at least $\ell$, it follows that $B_0$ is decoupled from $B_1, B_2, \cdots$. Thus, the union lemma applies to the entire collection $B_0,B_1,B_2,\cdots$, and it follows that $B = B_0 \cup B_1 \cup B_2 \cup \dots$ is correctable. 

    \item {\bf $\abs{C} < k$. } By our choice of grid tiling, the number of qubits in $C$ is at most 
    \begin{align}
    \frac{16D^2\ell^2}{w_0^2}n \le \frac{81}{100D}k,
    \end{align}
    which is the same inequality considered in~\eqref{eq:ineq1}. Again by our choice of tiling, the number of bad interactions with the qubits of $\mathcal{B}'$ is at most 
    \begin{align}
        \frac{8D\ell}{w_0} \cdot 2c_0 d  &\le \frac{8D\ell}{w_0}\cdot 2c_0\sqrt{kn} = 4c_0\left(\frac{4D\ell}{w_0}\sqrt{n}\right)\sqrt{k} \le 2c_0\cdot \frac{9}{10\sqrt{D}}k,
    \end{align}
    where the first inequality uses our assumption $d \le \sqrt{kn}$ and the last inequality follows from~\eqref{eq:ineq1}. Thus we see that 
    \begin{equation}
        \abs{C} \leq \frac{16D^2\ell^2}{w^2}n + \frac{8D\ell}{w}\cdot 2c_0d \le \frac{81}{100D}k+\frac{9c_0}{5\sqrt{D}}k < \frac{1}{2}k.
    \end{equation}
    Since $A,B$ are correctable, this gives us our desired contradiction using Lemma~\ref{lem:abc}.
\end{itemize}
We have obtained a contradiction in both the $d\le k$ and $d\ge k$ cases, and this complete the proof of the theorem.
\end{proof}

\section{Construction for Upper Bounds}
\label{sec:construction}

The bounds derived in Theorems~\ref{thm:main} and~\ref{thm:stab} are tight in both the interaction count $M^*$ and the interaction length $\ell^*$. Tightness is shown by constructing explicit examples of embedded codes which saturate the interaction count or length. For the interaction count, it suffices to consider an asymptotically good quantum low-density parity-check (qLDPC) code \cite{panteleev2021asymptotically,leverrier2022quantum}, which has $O(M^*)=O(\max(k,d))$ interactions of any length. Since a stabilizer code can also be regarded as a subsystem code with zero gauge qubits, this shows that both Theorem~\ref{thm:main} and~\ref{thm:stab} are tight in terms of interaction count. This is covered by Theorem 1.3 of~\cite{dai2024locality}.

We now describe constructions that show the interaction length is optimal in Theorem~\ref{thm:main} and Theorem~\ref{thm:stab}. In both cases, we construct a code that saturates the bound for interaction length by concatenating an asymptotically good qLDPC code with a geometrically local code which saturates the Bravyi and BPT bounds, respectively.

\subsection{Subsystem codes}\label{sec:subsys_const}
We start by showing the interaction length for subsystem codes (Theorem~\ref{thm:main}) is optimal. We will define a concatenated subsystem code composed of an asymptotically good qLDPC code, together with a subsystem code which is geometrically local in $D$-dimensions and saturates the Bravyi bound. For the local subsystem code, we employ the ``wire code'' construction of Baspin and Williamson~\cite{baspin2024wire}.

\begin{theorem}[Wire code \cite{baspin2024wire}]\label{thm:wire_codes}
For all $D\ge 2$, there exists an $\varepsilon>0$ such that, for all positive integers $n$ there exists a subsystsem code with parameters $[[n,\ge\varepsilon n^{\frac{D-1}{D}}, \ge\varepsilon n^{\frac{D-1}{D}}]]$ that has a set of gauge generators that are $O(1)$-local in a $D$-dimensional embedding.
\end{theorem}
\begin{proof}
    Apply the wire code construction of \cite{baspin2024wire} to a good qLDPC code (see Corollary 1 and Theorem 4 of \cite{baspin2024wire}).
    We justify that this gives a desired code not just infinitely many $n$, but for all $n$.
    Asymptotically good qLDPC codes exist for all $n$ (to see why it is all $n$, and not just infinitely many $n$, see discussion in \cite{dai2024locality}), so wire codes exist for sufficiently dense values of $n$ --- Theorem 4 of~\cite{baspin2024wire} constructs wire codes of length $O(n^\frac{D}{D-1})$ given a good qLDPC code of length $n$ as input. After padding with unused qubits and adjusting $\eps$, we get the desired local subsystem code family for all $n$.
\end{proof}

The concatenation procedure for subsystem codes is formally identical to the process for stabilizer codes. Namely, if $\calC_1 = [[n_1,k_1]]$ and $\calC_2 = [[n_2,k_2]]$ are subsystem codes, then their concatenation $\calC_2\circ\calC_1$ is defined using $n_2$ blocks of the inner code $\calS_1$ and $k_1$ copies of the outer code $\calC_2$. Let $q_{ij}$ be the $i$th logical qubit of the $j$th $\calS_1$ block. Then the concatenated code is defined by replacing the $j$th physical qubit of the $i$th $\calC_2$ block with $q_{ij}$. The properties of the concatenated code is summarized in Lemma~\ref{lem:subsystem_concatentation}.

\begin{lemma}[Concatenated Subsystem Codes]\label{lem:subsystem_concatentation}
Let $\calC_1 = [[n_1,k_1,d_1,g_1]]$ and $\calC_2 = [[n_2,k_2,d_2,g_2]]$ be two subsystem codes. Then there exists a subsystem code $\calC = \calC_2\circ\calC_1 = [[n_1n_2, k_1k_2, d\ge d_1d_2, k_1g_2 + n_2g_1]]$, called the concatenation of $\calC_2$ and $\calC_1$, such that:
\begin{enumerate}
    \item The gauge group $G$ of $\calC$ is generated by operators that are either:
    \begin{enumerate}
        \item\label{itm:gena} the tensor product of a gauge generator $g$ from one of the $n_2$ codeblocks of $\calC_1$ with the identity on all other code blocks,

        \item\label{itm:genb} an operator $\overline{g}$ formed by taking a gauge generator $g$ from the $j$th copy of $\calC_2$, and replacing the Pauli operator $P_i$ (acting on qubit $i$) from its tensor product decomposition with a corresponding bare logical Pauli operator $\overline{P}_{ij}$ for the $j$th logical qubit encoded in $i$th encoded $\calC_1$ codeblock.\footnote{With the inclusion of operators from part (a), we can equivalently take a dressed logical Pauli rather than the bare Pauli.}
    \end{enumerate}

    \item\label{itm:stab} The stabilizer group $S$ of $\calC$ is generated by operators that are either:
        \begin{enumerate}
        \item the tensor product of a stabilizer generator $M$ from one of the $n_2$ codeblocks of $\calC_1$ with the identity on all other code blocks,

        \item an operator $\overline{M}$ formed by taking a stabilizer generator $M$ from the $j$th copy of $\calC_2$, and replacing the Pauli operator $P_i$ (acting on qubit $i$) from its tensor product decomposition with a corresponding bare logical Pauli operator $\overline{P}_{ij}$ for the $j$th logical qubit encoded in $i$th encoded $\calC_1$ codeblock.
    \end{enumerate}
\end{enumerate}
\end{lemma}
\begin{proof}
Let the generator groups of $\calC_1$ and $\calC_2$ be $G_1$ and $G_2$, respectively. Note that type~\eqref{itm:gena} gauge generators in $G$ are in bijection with the generators of $G_1$. Likewise, type~\eqref{itm:genb} generators in $G$ are in bijection with the generators of $G_2$. Moreover, any two type~\eqref{itm:gena} (resp. type~\eqref{itm:genb}) generators retain the same commutation relations from their defining operators. Since every bare logical operator commutes with all gauge operators, it follows that all type~\eqref{itm:genb} generators commute with all type~\eqref{itm:gena} generators. The discussion above implies that there exists a set of generators for $G$ of the form
\begin{align}
G = \langle S,\overline{X}_1,\overline{Z}_1,\cdots,\overline{X}_g,\overline{Z}_g\rangle,
\end{align}
where $S$ is the stabilizer group defined in item~\ref{itm:stab}, and where $\overline{X}_i,\overline{Z}_i$ are logical Paulis with the usual commutation relations. It follows from the form of these generators that $G$ is a well-defined gauge group for a subsystem code. This proves that the concatenated code $\calC$ is well-defined.

The parameters $k=k_1k_2$ and $g=k_1g_2+n_2g_1$ of the concatenated code follow from a simple counting of the generators. For the distance, note that any error on $\calC$ must induce at least $d_2$ errors on the physical level of $\calS_2$. By the construction of the concatenated code, this is equivalent to logical errors on the associated blocks of $\calS_1$. The total weight of such an error is therefore at least $d_1d_2$.
\end{proof}

\begin{theorem}[Optimality of Interaction Length for Subsystem Codes]
    For all $D\ge 2$, there exists a constant $c_1 = c_1(D) > 0$ such that the following holds: for all $n,k,d > 0$ with $k,d\le n$ satisfying $kd^{\frac{1}{D-1}} \geq c_1 n$ or $d\ge c_1 n^{\frac{D-1}{D}}$, there exists an $[[n, \ge k, \ge d]]$ subsystem code with a $D$-dimensional embedding containing no interactions of length at least 
    \begin{align}
    \ell = \max\bigg(\frac{d}{n^\frac{D-1}{D}}, \bigg(\frac{kd^\frac{1}{D-1}}{n}\bigg)^\frac{D-1}{D}\bigg).
    \end{align}
    \label{thm:constr-1}
\end{theorem}

\begin{proof}

We will prove the following claim: \emph{for all $D \ge 2$, there exists constants $\eps \in (0,1)$ and $c'_1 > 0$ such that for all $n,k,d > 0$ with $k,d\le n$ satisfying $kd^{\frac{1}{D-1}} \geq c_1'n$ or $d \geq c_1'n^\frac{D-1}{D}$, there exists a $[[n, \geq \varepsilon k, \geq \varepsilon d]]$ subsystem code with no interactions of length at least 
\begin{align}
\ell = \max\bigg(\frac{d}{n^\frac{D-1}{D}}, \bigg(\frac{kd^\frac{1}{D-1}}{n}\bigg)^\frac{D-1}{D}\bigg).
\end{align}}
The theorem follows by taking $c_1=c_1'/\varepsilon^2$ and applying the above claim with $k\mapsto\ceil{k/\varepsilon}$ and $d\mapsto\ceil{d/\varepsilon}$.

Note that it suffices to consider the case $d \le k$, where we will construct subsystem codes with no interactions of length at least $\ell = (\frac{kd^{\frac{1}{D-1}}}{n})^{\frac{D-1}{D}}$. Given $d \ge k$, we can always reduce to the $d \le k$ case by applying the construction using parameters $n,k',d$, with $k'=d \ge k$, so that we get a $[[n,\ge k,d]]$ code with no interactions of length at least
\begin{align}
\ell=\frac{d}{n^{\frac{D-1}{D}}}=\bigg(\frac{k'd^{\frac{1}{D-1}}}{n}\bigg)^{\frac{D-1}{D}}.
\end{align}

Now, fix parameters $n,k,d>0$ satisfying the hypothesis of the claim and with $d \le k$. Let $\varepsilon_1\in(0,1)$ be such that there exists good qLDPC codes with parameters $[[n_1,\ge\varepsilon_1 n_1,\ge\varepsilon_1 n_1]]$ for all $n_1 \in \bbN$. Similarly, let $\varepsilon_2\in(0,1)$ and $\ell_2=O(1)$ be such that there exists wire codes with parameters $[[n_2, \ge\varepsilon_2 n_2^{(D-1)/D}, \ge\varepsilon_2 n_2^{(D-1)/D}]]$ and locality $\ell_2$ in $D$ dimensions, for all $n_2\in 
\bbN$. That this is possible is a consequence of Theorem~\ref{thm:wire_codes}.
Let $\varepsilon=\varepsilon_1\varepsilon_2/\ell'$ where $\ell'=2(\sqrt{D}+\ell_2)$, and let $\ell = (\frac{kd^\frac{1}{D-1}}{n})^\frac{D-1}{D}$. Set $n_0 = nd/k$, $n_1=(\ell/\ell')^D$, and $n_2 = n_0/n_1$.

We will ignore rounding errors from non-integer parameters; all parameters that ``should'' be integer, i.e., $n_0$, $n_1$, and $n_2$, are bounded away from zero by some constant $c=c(D,\ell_2)>0$.
Rounding these values to integers incurs the cost of at most another constant factor (dependent on $c$) to the parameter $\varepsilon$. It suffices to take $c_1' = (\ell')^2$ to ensure that $\ell \ge \ell'$, so that all resulting quantities are well-defined.

Let $\calC_1$ be a good qLDPC code with $n_1=(\ell/\ell')^D$ qubits and parameters $[[n_1, \ge \eps_1n_1, \ge \eps_1n_1]]$. Let $\calC_2$ be a wire code with $n_2 = n_0/n_1$ qubits and parameters $[[n_2, \geq \varepsilon_2 n_2^{(D-1)/D}, \geq \varepsilon_2 n_2^{(D-1)/D}]]$. Perform the subsystem code concatenation process (see Lemma~\ref{lem:subsystem_concatentation}) with inner code $\calC_1$ and outer code $\calC_2$. This gives us a code $\calC_0$ using $n_0 = n_1n_2 = nd/k$ qubits with dimension and distance at least
\begin{align}
    \eps_1n_1\cdot \eps_2n_2^\frac{D-1}{D} = \varepsilon_1\varepsilon_2 n_0^\frac{D-1}{D} \left(\frac{\ell}{\ell'}\right).
\end{align}
Now, let $\calC$ be the code obtained from taking $\ceil{n/n_0}$ disjoint copies of $\calC_0$. Then $\calC$ has dimension at least
\begin{align}
    \left\lceil\frac{n}{n_0}\right\rceil \cdot \varepsilon_1\varepsilon_2 n_0^\frac{D-1}{D} \left(\frac{\ell}{\ell'}\right) \geq \frac{\varepsilon_1\varepsilon_2}{\ell'}\cdot \frac{n\ell}{n_0^{1/D}} = \eps\left(\frac{nd}{n_0}\right)^\frac{1}{D} k^\frac{D-1}{D} = \varepsilon k.
\end{align}
The distance of $\calC$ is the same as the distance of $\mathcal{C}_0$, which is at least
\begin{align}
    \varepsilon_1\varepsilon_2 n_0^\frac{D-1}{D} \left(\frac{\ell}{\ell'}\right) = \frac{\varepsilon_1\varepsilon_2}{\ell'} \cdot n_0^\frac{D-1}{D}\ell=\eps \left(\frac{n_0k}{n}\right)^\frac{D-1}{D}d^\frac{1}{D} = \varepsilon d.
\end{align}

It follows that $\calC$ has parameters $[[n, \geq \varepsilon k, \geq \varepsilon d]]$, as required,

Finally, we now exhibit a $D$-dimensional embedding of $\calC$ with no interactions of length at least $\ell$. First, note that the qubits of the good qLDPC code $\calC_1$ can be embedded into a cubic lattice of side lengths at most $n_1^{1/D} = \ell/\ell'$. Any interaction between qubits in such an embedding of $\calC_1$ has length at most $(\ell/\ell')\sqrt{D}$. The desired embedding for $\calC$ follows from the intrinsic local embedding of the wire code $\calC_2$ (with locality $\ell_2$), but with each qubit replaced by a $\calC_1$ block embedded as a cubic lattice. With this replacement, we also dilate the embedding for $\calC_2$ by a factor of $\ell/\ell'$, so that interacting $\calC_1$ blocks are at a distance of at most $(\ell/\ell')\cdot \ell_2$ apart (center-to-center). It follows that the maximum interaction length between individual qubits is at most
\begin{align}
(\ell/\ell')\cdot \ell_2 + (\ell/\ell')\sqrt{D} = (\ell/\ell')(\ell_2+\sqrt{D}) = \ell/2 < \ell,
\end{align}
where the first expression is the sum of the inter- and intra-block lengths. It follows that $\calC$ admits a $D$-dimensional embedding with no interactions of length $\ge \ell$.

\end{proof}

\subsection{Commuting Projector Codes}\label{sec:stab_const}
We now show the interaction length in Theorem~\ref{thm:stab} is optimal. The construction is very similar to the one used in Theorem 1.3 of~\cite{dai2024locality}, except generalized to $D$-dimensions. In 2D, the surface code offers a simple and natural candidate for a geometrically local code that saturates the BPT bound. In higher dimensions, we instead use the family of ``subdivided codes'' constructed by Lin, Wills and Hsieh~\cite{li2024transform}. The optimality of the interaction length follows by concatenating a good qLDPC code with the subdivided code.

\begin{theorem}[Subdivided code \cite{li2024transform}]\label{thm:subdivided}
For all $D\ge 2$, there exists an $\varepsilon>0$ such that, for all positive integers $n$ there exists a stabilizer code with parameters $[[n,\ge\varepsilon n^{\frac{D-2}{D}}, \ge\varepsilon n^{\frac{D-1}{D}}]]$ that has a set of stabilizer generators that are $O(1)$-local in a $D$-dimensional embedding.
\end{theorem}
\begin{proof}
The $D=2$ is handled by the surface code. For $D\ge 3$, apply the subdivided code construction of \cite{li2024transform} to a good qLDPC code (see Theorem 5.1 of \cite{li2024transform}). We justify that this gives a desired code not just infinitely many $n$, but for all $n$. Asymptotically good qLDPC codes exist for all $n$ \cite{panteleev2021asymptotically} (to see why it is all $n$, and not just infinitely many $n$, see the proof of Theorem 1.3 in \cite{dai2024locality}), so subdivided codes exist for sufficiently dense values of $n$ --- Theorem 5.1 of \cite{li2024transform} gives a subdivided code of length $O(n^{\frac{D}{D-2}})$ from a good qLDPC code of length $n$. After padding with unused qubits and adjusting $\eps$, the desired family of local stabilizer codes exist for all $n$.
\end{proof}

\begin{theorem}[Optimality of Interaction Length for Stabilizer Codes] 
For all $D\ge 2$, there exists a constant $c_1=c_1(D)>0$ such that the following holds: for all $n,k,d > 0$ with $k,d \leq n$ satisfying either $kd^{\frac{2}{D-1}} \geq c_1\cdot n$ or $d\ge c_1\cdot n^{\frac{D-1}{D}}$, there exists a $[[n, \geq k, \geq d]]$ quantum stabilizer code with a $D$-dimensional embedding containing no interactions of length at least 
\begin{align}
\ell = \max\bigg(\frac{d}{n^{\frac{D-1}{D}}},  \bigg(\frac{kd^{\frac{2}{D-1}}}{n}\bigg)^{\frac{D-1}{2D}}\bigg).
\end{align}
\label{thm:constr-2}
\end{theorem}

\begin{proof} 
We prove the following claim: \emph{for all $D\ge 2$, there exists constants $\varepsilon \in (0,1)$ and $c'_1 > 0$ such that for all $n,k,d > 0$ with $k,d \leq n$ satisfying either $kd^{\frac{2}{D-1}} \geq c_1'n$ or $d \ge c_1'n^{\frac{D-1}{D}}$, there exists a $[[n, \geq \varepsilon k, \geq \varepsilon d]]$ quantum stabilizer code with a $D$-dimensional embedding with no interactions of length at least 
\begin{align}
\ell = \max\bigg(\frac{d}{n^{\frac{D-1}{D}}}, \bigg(\frac{kd^{\frac{2}{D-1}}}{n}\bigg)^{\frac{D-1}{2D}}\bigg).
\end{align}
}
The theorem follows from taking $c_1=c_1'/\varepsilon^3$ and applying the above claim with $k \mapsto \ceil{k/\varepsilon}$ and $d\mapsto\ceil{d/\varepsilon}$.

Note that it suffices to consider the case $d \le {kn}$, where we will construct a stabilizer codes with no interactions of length $\ell = \big(\frac{kd^{\frac{2}{D-1}}}{n}\big)^{\frac{D-1}{2D}}$. Given $d \geq \sqrt{kn}$, we can reduce to the $d \le \sqrt{kn}$ case by applying the construction using parameters $n,k',d$, where $k'\ge k$ satisfies $d = \sqrt{k'n}$, so that we get a $[[n,\ge k,d]]$ code with no interactions of length at least
\begin{align}
\ell=\frac{d}{n^{\frac{D-1}{D}}}=\bigg(\frac{k'd^{\frac{2}{D-1}}}{n}\bigg)^{\frac{D-1}{2D}}.
\end{align}
Now, fix parameters $n,k,d > 0$ satisfying the hypothesis of the claim and with $d \le \sqrt{kn}$. Let $\varepsilon_1\in(0,1)$ be such that there exists good qLDPC codes with parameters $[[n_1,\ge\varepsilon_1 n_1,\ge\varepsilon_1 n_1]]$ for all $n_1\in \bbN$. Similarly, let $\varepsilon_2\in(0,1)$ and $\ell_2=O(1)$ be such that there exists subdivided codes with parameters $[[n_2,\ge\varepsilon_2n_2^{\frac{D-2}{D}},\ge\varepsilon_2 n_2^{\frac{D-1}{D}}]]$ and locality $\ell_2$ in $D$ dimensions, for all $n_2\in \bbN$. That this is possible is a consequence of Theorem~\ref{thm:subdivided}. Let $\varepsilon=\varepsilon_1\varepsilon_2/(\ell')^2$, where $\ell'=2(\sqrt{D}+\ell_2)$, and let $\ell = \big(\frac{kd^{\frac{2}{D-1}}}{n}\big)^{\frac{D-1}{2D}}$. Set $n_0 = (d/\ell)^{\frac{D}{D-1}}$, $n_1 = (\ell/\ell')^D$, and $n_2 = n_0/n_1$.

We ignore rounding errors from non-integer parameters.
All parameters that ``should'' be integers, i.e., $n_0$, $n_1$, $n_2$, and $n/n_0$, are bounded away from zero by some constant $c=c(D,\ell_2)>0$. Rounding these values to integers incurs the cost of at most a constant factor (that depends on $c$) to the parameter $\varepsilon$. It suffices to take $c_1' = (\ell')^3$ to ensure that $\ell \ge \ell'$, so that all resulting quantities are well-defined.

Let $\mathcal{S}_1$ be a good qLDPC code with $n_1=(\ell/\ell')^D$ qubits and parameters $[[n_1, \geq \varepsilon_1 n_1, \geq \varepsilon_1 n_1]]$. Let $\mathcal{S}_2$ be a subdivided code with $n_2=n_0/n_1$ qubits and parameters $[[n_2, \geq \varepsilon_2 n_2^{(D-2)/D}, \geq \varepsilon_2 n_2^{(D-1)/D}]]$. Let $\calS_0$ be the stabilizer code obtained from the concatenation of $\calS_1$ (inner code) and $\calS_2$ (outer code). The code $\calS_0$ has $n_0=n_1n_2 = (d/\ell)^\frac{D}{D-1}$ qubits. The distance of $\calS_0$ is at least
\begin{align}
\eps_1\eps_2n_1n_2^\frac{D-1}{D} = \eps_1\eps_2n_1\left(\frac{n_0}{n_1}\right)^\frac{D-1}{D} = \frac{\eps_1\eps_2}{\ell'}d \ge \eps d,
\end{align}
and the dimension of $\calS_0$ is at least
\begin{align}
\eps_1\eps_2n_1n_2^\frac{D-2}{D} = \eps_1\eps_2n_1\left(\frac{n_0}{n_1}\right)^\frac{D-2}{D} = \eps_1\eps_2n_1 \left(\frac{d}{\ell}\right)^{\frac{D-2}{D-1}} \left(\frac{\ell'}{\ell}\right)^{D-2} = \frac{\eps_1\eps_2}{(\ell')^2}\ell^\frac{D}{D-1} d^\frac{D-2}{D-1} = \eps\ell^\frac{D}{D-1} d^\frac{D-2}{D-1}.
\end{align}
Now, let $\calS$ be the code obtained from taking $\ceil{n/n_0}$ disjoint copies of $\mathcal{S}_0$. The distance of $\calS$ is the same as the distance of $\calS_0$, which is $\ge \eps d$. The dimension of $\calS$ is at least
\begin{align}
    \left\lceil\frac{n}{n_0}\right\rceil\cdot \eps\ell^\frac{D}{D-1} d^\frac{D-2}{D-1} \ge \frac{n}{n_0}\cdot \eps\ell^\frac{D}{D-1} d^\frac{D-2}{D-1} = \eps \frac{n}{d^\frac{2}{D-1}} \ell^\frac{2D}{D-1} = k\eps.
\end{align}
It follows that $\mathcal{S}$ has parameters $[[n,\ge\eps k,\ge \eps d]]$, as required. 

Finally, we now exhibit a $D$-dimensional embedding of $\calS$ with no interactions of length at least $\ell$. First, note that the qubits of the good qLDPC code $\calS_1$ can be embedded into a cubic lattice of side lengths at most $n_1^{1/D} \le \ell/\ell'$. Any interaction between qubits in such an embedding of $\calS_1$ has length at most $(\ell/\ell')\sqrt{D}$. The desired embedding for $\calS$ follows from the intrinsic local embedding (with locality $\ell_2$) of the subdivided code $\calS_2$, but with each qubit replaced by a $\calS_1$ block embedded as a cubic lattice. With this replacement, we also dilate the embedding for $\calS_2$ by a factor of $\ell/\ell'$, so that interacting $\calC_1$ blocks are at a distance of at most $(\ell/\ell')\cdot \ell_2$ apart (center-to-center). It follows that the maximum interaction length between individual qubits is at most
\begin{align}
(\ell/\ell')\cdot \ell_2 + (\ell/\ell')\sqrt{D} = (\ell/\ell')(\ell_2+\sqrt{D}) = \ell/2 < \ell,
\end{align}
where the first expression is the sum of the inter- and intra-block lengths. It follows that $\calS$ admits a $D$-dimensional embedding with no interactions of length $\ge \ell$.

\end{proof}

\section*{Acknowledgements}
We thank Andy Liu, Shouzhen Gu, Zhiyang He for helpful feedback on our manuscript.
RL is supported by NSF grant CCF-2347371.

\bibliographystyle{alpha}
\bibliography{bib}

\end{document}